\documentclass[journal]{IEEEtran}
\ifCLASSINFOpdf
\else
\fi

\usepackage{cite}
\usepackage{mathtools}
\usepackage{amsmath, amsthm, amssymb,nccmath}
\usepackage{algorithm}
\usepackage[noend]{algpseudocode}
\usepackage{multirow}
\usepackage{makecell}
\usepackage{float}
\usepackage{subfig}
\usepackage{booktabs}

\let\E\relax
\let\U\relax
\let\G\relax
\let\define\relax
\DeclareMathOperator{\G}{\mathcal{G}}

\DeclareMathOperator{\R}{\mathcal{R}}
\DeclareMathOperator{\E}{\mathbb{E}}
\DeclareMathOperator{\I}{\mathbb{I}}

\DeclareMathOperator{\U}{\mathcal{U}}
\DeclareMathOperator{\Dis}{\mathcal{D}}
\DeclareMathOperator{\gr}{Greedy}

\DeclareMathOperator{\dA}{\dot{A}}
\DeclareMathOperator{\dD}{\dot{D}}
\DeclareMathOperator{\dL}{\dot{L}}

\DeclareMathOperator*{\argmax}{arg\,max}
\DeclareMathOperator{\define}{\coloneqq}
\DeclareMathOperator{\neqZ}{\mathbb{Z}_{\geq 0}}
\DeclareMathOperator{\Ind}{Indi}
\DeclareMathOperator{\Ma}{M_a}
\DeclareMathOperator{\Mt}{M_t}
\DeclareMathOperator{\ifadd}{IfAdd}
\DeclareMathOperator{\true}{True}
\DeclareMathOperator{\false}{False}
\DeclareMathOperator{\InDeg}{InDeg}

\theoremstyle{definition}
\newtheorem{problem}{\textbf{Problem}}
\newtheorem{remark}{\textbf{Remark}}

\newtheorem{policy}{\textbf{Policy}}
\newtheorem{lemma}{\textbf{Lemma}}

\newtheorem{definition}{Definition}

\newcommand{\RNum}[1]{\uppercase\expandafter{\romannumeral #1\relax}}

\begin{document}
%
\title{Time-constrained Adaptive Influence Maximization}
%
%
%

\author{Guangmo~Tong,~\IEEEmembership{Member,~IEEE,}
        Ruiqi~Wang,
        Zheng~Dong, ~\IEEEmembership{Member,~IEEE,}
        and~Xiang~Li,~\IEEEmembership{Member,~IEEE}
\thanks{G. Tong and R. Wang are with the Department
of Computer and Information Sciences, University of Delaware, E-mail: \{amotong, wangrq\}@udel.edu.}
\thanks{Z. Dong is with Department of Computer Science, Wayne State University, E-mail:  dong@wayne.edu.}
\thanks{X. Li is with Department of Computer Science and Engineering, Santa Clara University, E-mail: xli8@scu.edu.}}

%
%

\markboth{Journal of \LaTeX\ Class Files,~Vol.~14, No.~8, August~2015}%
{Shell \MakeLowercase{\textit{et al.}}: Bare Demo of IEEEtran.cls for IEEE Journals}
%



\maketitle

\begin{abstract}
The well-known influence maximization problem aims at maximizing the influence of one information cascade in a social network by selecting appropriate seed users prior to the diffusion process. In its adaptive version, additional seed users can be selected after observing certain diffusion results. On the other hand, social computing tasks are often time-critical, and therefore only the influence resulted in the early period is worthwhile, which can be naturally modeled by enforcing a time constraint. In this paper, we present an analysis of the \textit{time-constrained adaptive influence maximization} problem. On the theory side, we provide the hardness results of computing the optimal policy and a lower bound on the adaptive gap. For practical solutions, from basic to advanced, we design a series of seeding policies for achieving high efficacy and scalability. Finally, we investigate the proposed solutions through extensive simulations based on real-world datasets. 
\end{abstract}


%
\IEEEpeerreviewmaketitle

%
%
%
%
 \section{Introduction}
\IEEEPARstart{A}{s} one of the core research branches in social network analysis, influence maximization (\textbf{IM}), proposed by Kempe, Kleinberg and Tardos \cite{kempe2003maximizing}, studies the problem of launching information cascades such that the influence can be maximized. Inspired by influence maximization, various topics in online social networks have been investigated, such as misinformation control, online friending, and viral marketing \cite{li2018influence,zhang2014recent,sun2011survey,aslay2018influence}. The classic influence maximization problem adopts two settings: (a) \textit{non-adaptive strategy}: the seed users are all computed before the diffusion process, and (b) \textit{unlimited time steps}: the influence is counted without a time limit. These classic settings are elegant, but they are incapable of modeling many real applications. First, in order to optimize the seeding selection, one often prefers to deploy seed nodes adaptively, which is formulated as the Adaptive Influence Maximization (\textbf{AIM}) problem \cite{golovin2010adaptive,tong2017adaptive,vaswani2016adaptive,han2018efficient,chen2019adaptivity}. Allowing an adaptive seeding enables us to identify the best seed node(s) conditioned on the observed diffusion results, and it therefore can result in a higher influence under the budget constraint. For example, a higher profit would be expected if our online advertisements were posted adapted to customer feedback \cite{kazienko2007adrosa}. Second, time-critical applications are commonly seen in online social networks, and in such cases, only the influence resulted before the deadline matters. For instance, launching a positive cascade to counter misinformation is expected to exert effects expeditiously \cite{farajtabar2017fake,tong2018misinformation}. For such scenarios, we would like to maximize the influence under a time constraint, which is termed as the Time-constrained Influence Maximization (\textbf{TIM}) problem \cite{liu2012time,liu2013influence,han2017time,xie2015dynadiffuse,chen2012time}.  In order to support time-critical tasks through adaptive seeding methods, we in this paper propose the Time-constrained Adaptive Influence Maximization (\textbf{TAIM}) problem.

\textbf{Problem Formulation.} An adaptive seeding process alternates between seeding steps and diffusing steps, and in each seeding step, we select a set of seed users to trigger more influence, and our decision is made adapted to the observed diffusion results. An adaptive seeding policy essentially consists of two modules, \textit{seeding pattern} and \textit{node selection rule}, where the seeding pattern specifies the size of the seed set while the node selection rule determines which nodes to select.  Given two integers $K, T \in \mathbb{Z}^+$, the TAIM problem asks for a policy to deploy $K$ seed nodes in an adaptive manner such that the total influence resulted in the first $T$ diffusion rounds can be maximized. In this paper, we study the TAIM problem and aim at both the theoretical analysis and practical solution design. 

\textbf{A Key Trade-off.} The TAIM problem is a natural combination of AIM and TIM, both of which have been extensively studied and have been shown to admit the ($1-1/e$)-approximation subject to controllable sampling errors. For the AIM problem, the optimal seeding policy follows the full-adoption feedback model \cite{golovin2010adaptive} in which (a) before making the next seeding decision, we always keep observing the diffusion process until it terminates, and (b) we always use one budget whenever a seeding action has to be performed. Such a seeding pattern is intuitively optimal as it maximally obtains observations before selecting the next seed node. However, when a time constraint is enforced, one can see that the full-adoption feedback model is not optimal anymore, because waiting for more diffusion rounds, though brings more feedback, will incur the loss of future diffusion rounds. In short, waiting is not ``free'' in TAIM. Consequently, the critical issue is to determine the balance between (a) waiting for more feedback and (b) performing a seeding action at an early stage. We observe that solving such a trade-off in optimal is theoretically hard, making TAIM different from the existing problems. Through appropriate methods designed in this paper for achieving a reasonable trade-off, we have been able to design seeding policies that can solve the TAIM problem effectively.

\textbf{Contributions.} This paper presents a systematic analysis of the TAIM problem, and the contributions are briefly summarized as follows:
\begin{itemize}
	\item We perceive the adaptive seeding process as a procedure alternating between seeding steps and diffusing steps, based on which we propose the Time-constrained Adaptive Influence Maximization (TAIM) problem, which finds the policy to compute a seed set in each seeding step subject to a budget constraint such that the influence within a time limit can be maximized. 
	\item Theoretically, we prove that TAIM problem exhibits a unique hardness that is different from existing problems such as IM or AIM. Furthermore, we provide the first result on the adaptive gap for the time-constrained case and prove a lower bound of $\frac{e^2-2}{e-1}$. 
	\item Towards solving TAIM effectively, we design a sampling method to enable an efficient greedy node selection rule for the time-constrained case, based on which we propose a collection of seeding policies, from basic to advanced, including static seeding policy, greedy seeding policy, and several foresight seeding policies. We experimentally evaluate the proposed polices through simulations on real-world graphs, in terms of effectiveness, efficiency and robustness. As a minor part, we contribute a new Reddit dataset for studying information diffusion. Our source code and data will be made online available.
\end{itemize}

\textbf{Roadmap.} The related work will be introduced in Sec. \ref{sec: relate}. We provide the preliminaries in Sec. \ref{sec: pre}, including the diffusion model and the formulation of the TAIM problem. The theoretical analysis is given in Sec. \ref{sec: theory}, and the designed seeding policies are then described in Sec. \ref{sec: strategy}. In Sec. \ref{sec: exp}, we present the experimental study. Sec. \ref{sec: con} concludes. 

\section{Related Work}
\label{sec: relate}
Influence maximization and its variants have been extensively studied. In this section, we survey the works germane to our work. 

\textbf{IM, TIM, and AIM.} The IM problem \cite{kempe2003maximizing} investigates the strategy to launch an information cascade in social networks, with the goal of maximizing the resulted influence. It has been proved that the IM problem is monotone and submodular under the classic diffusion models (e.g., independent cascade model and linear threshold model), and therefore a $(1-1/e)$-approximation can be readily obtained by the greedy strategy due to the celebrated results of Nemhauser \textit{et al.} \cite{nemhauser1978analysis}. However, the objective function (i.e., influence) of the IM problem is \#$P$-hard to compute, so efficient heuristics were designed by various methods (e.g., \cite{chen2009efficient,chen2010scalable}). Borgs \textit{et al.} \cite{borgs2014maximizing} later invented the reverse sampling technique resulting in an efficient algorithm without sacrificing the performance guarantees. The reverse sampling technique was further improved by a series of works \cite{tang2014influence,tang2015influence,nguyen2016stop}, and currently, the IM problem can be solved efficiently on even very large networks. In order to support time-critical applications, researchers have further considered the IM problem with a time constraint \cite{liu2012time, liu2013influence,han2017time,xie2015dynadiffuse,chen2012time,dinh2013cost}. The TIM problem remains monotone and submodular, so the greedy algorithm still gives an effective approximation solution. Because the diffusion process is stochastic, it is possible to adopt an adaptive seeding policy where we could compute the seed nodes after observing the diffusion feedback, which was first considered by Golovin \textit{et al.} \cite{golovin2011adaptive} using the technique of adaptive submodularity. Under the budget constraint, a non-adaptive seeding policy computes a subset of nodes with a specified size, while an adaptive seeding policy computes a seed set in each seeding step according to the observations subject to the budget constraint. Without a time constraint, it has been shown that the full-adoption feedback model \cite{golovin2010adaptive} combined with the greedy node selection rule would give a $(1-1/e)$-approximation for the AIM problem under the budget constraint \cite{tong2017adaptive}. However, with a time constraint, the problem is not adaptive submodular \cite{vaswani2016adaptive,golovin2011adaptive}, and thus the current technique cannot be applied. Vaswani \textit{et al.} \cite{vaswani2016adaptive} (Arxiv.org) considered both the AIM and TAIM problem\footnote{They called time constraint as bounded time horizon.} and suggested using sequential model based optimization (SMBO) \cite{hutter2011sequential} to deal with the general case. While their ideas are intuitive, they did not provide the detailed implementation and their experiments for TAIM focused on examining the average adaptivity gain but not the efficacy in solving TAIM. Other works have studied the AIM problem under specific feedback models \cite{sun2018multi,salha2018adaptive,seeman2013adaptive, horel2015scalable,tong2019adaptive} or considered the trade-off between delay and efficiency \cite{tang2020influence, stein2017heuristic} based on partial feedback, but their settings still allow the diffusion process to complete and therefore cannot meet a hard time constraint. 



\textbf{Adaptive Gap.} Another important concept is the adaptive gap, which measures the ratio between the optimal adaptive policy and the optimal non-adaptive policy. Despite the recent results in \cite{peng2019adaptive,chen2019adaptivity,fujii2019beyond}, the adaptive gap under most of the feedback models is still open. In this paper, we provide the first result on the lower bound for the time-constrained case.

\section{Preliminaries}
\label{sec: pre}
This section provides the preliminaries to the rest of the paper. We first describe the considered diffusion model and then present the formulation of the TAIM problem.

\subsection{Diffusion Model}
We consider the classic independent cascade model in which a social network is given by a directed graph $G=(V, E)$, and associated with each edge $e \in E$ there is a propagation probability $p_e \in (0,1]$. We use $n \in \neqZ$ and $m \in \neqZ$ to denote the number of nodes and edges, respectively. A cascade is triggered by the seed users who are \textit{active} after selected. When a user $u$ becomes active, they have one chance to activate each inactive neighbor $v$ with a success probability of $p_{(u,v)}$.\footnote{For the purpose of analysis, we assume that $p_{(u,v)}$ is positive as we can remove the edges with zero propagation probability.} We assume that the diffusion process goes round by round.
\begin{definition}[\textbf{Round}]
	In each diffusion \textit{round}, the users, who are activated either by their neighbors in the last diffusion round or by being selected as new seed users, attempt to activate their inactive neighbors. 
\end{definition}
The diffusion process can be viewed as a stochastic BFS process. In this paper, the diffusion time is measured by the number of rounds.

\subsection{Seeding Process and Policy}
A seeding process consists of seeding steps and diffusing steps. In a seeding step, we can observe the activation results in the previous diffusion rounds, which can be equivalently represented by the states of the edges. In particular, we say the edge $(u,v)$ is \textit{live} if $u$ can successfully activate $v$. Otherwise, we say it is \textit{dead}. Therefore, an intermediate stage of a diffusion process is deductively determined by the current active nodes and the sets of live and dead edges. We introduce the concept of status for such a purpose.
\begin{definition}[\textbf{Status}]
	A \textit{status} $U=(\dA(U), \dL(U), \dD(U)) \in 2^V\times 2^E \times 2^E$ is given by a three-tuple with $\dL(U) \cap \dD(U) = \emptyset$, where $\dA(U)$ denotes the set of current active nodes, $\dL(U)$ and $\dD(U)$ are, respectively, the sets of live edges and dead edges. An edge is not observed yet iff it is not in  $\dL(U) \cup \dD(U)$. We use $\Phi$ to denote the status space.
\end{definition}

We employ the next concept to describe the scenario when the diffusion process terminates spontaneously.

\begin{definition}[\textbf{Final Status}]
	A status $U=(\dA(U), \dL(U), \dD(U)) $ is \textit{final} if all the edges from  $\dA(U)$ to $V\setminus \dA(U)$ are dead. That is, $\{(u,v)\in E: u \in \dA(U),~ v \in V\setminus \dA(U)\} \subseteq \dD(U)$, which implies that no node can be further activated unless new seed nodes are selected.
\end{definition}

\begin{definition}[\textbf{State}]
	We use $(U,t,k) \in \Phi \times \neqZ \times \neqZ$ to denote a \textit{state} of the seeding process, implying that the current status is $U$, the number of remaining diffusion rounds is $t$, and the remaining budget is $k$. 
\end{definition}

\begin{definition}[\textbf{Policy}]
	Given a state $(U,t,k)$  in a seeding step, a \textit{policy} $\pi$ computes a seed set $\pi(U,k,t) \subseteq V$ to be selected with $|\pi(U,k,t)|\leq k$. A policy  $\pi$ is \textit{non-adaptive} if it has $|\pi(U,k,t)|=k$ for each state $(U,k,t)$.
\end{definition}

\begin{definition}[\textbf{Seeding Process}]
	For a diffusion model with a time constraint $T\in \neqZ$, and a budget $K \in \neqZ$, the \textit{seeding process} under a policy $\pi$ is described as follows:
	\begin{itemize}
		\item Set $U=(\emptyset,\emptyset,\emptyset), k=K$ and $t=T$. Iterate the following process for $T$ times.
		\begin{itemize}
			\item \textbf{Seeding Step.} Compute and launch the seed set $\pi(U, k, t)$. Set $k=k-|\pi(U, k, t)|$.
			\item \textbf{Diffusing Step.} Observe the diffusion process for one round. Set $t=t-1$ and update $U$ as the observed status.
		\end{itemize}  
		\item Output the influence (i.e., the number of active nodes).
	\end{itemize}
\end{definition}

We use $f(\pi, K, T)$ to denote the expected influence associated with a policy $\pi$.  For a non-adaptive policy that selects a particular set $S \subseteq V$ as the seed nodes, we denote the resulted influence as $f(S, K, T)$.

\begin{remark}
	It is possible that no seed node is selected in a certain seeding step (i.e., $ \pi(U, k, t)=\emptyset$), which means that the policy would wait for more diffusion rounds. 
\end{remark}

\subsection{Problem Formulation}
The problem considered in this paper is stated below.

\begin{problem}[\textbf{TAIM Problem}]
	\label{problem: t-aim}
	Given a diffusion model, a time constraint $T\in \neqZ$, and a budget $K \in \neqZ$, design a policy $\pi$ such that $f(\pi, K, T)$ is maximized.
\end{problem}

\begin{remark}[\textbf{Special Cases}]
	\label{remark: relate}
	The TAIM problem is closely related to several problems that have been considered in the existing literature. 
	
	\begin{itemize}
		\item When $T=\infty$, it is exactly the unconstrained AIM problem, which admits a $(1-1/e)$-approximation achieved by combining the full-adoption feedback model and the greedy node selection rule. 
		
		\item When we have $p_e=1$ for each edge $e \in E$ (i.e., the deterministic model) or we have $T=1$, the optimal policy must be non-adaptive and therefore the greedy node selection rule provides a $(1-1/e)$-approximation. 
	\end{itemize} 
	A policy for the TAIM problem, if not optimal, should ideally provide the best possible solution when applied to those special cases. 
\end{remark}

\section{Theoretical Analysis}
\label{sec: theory}
\subsection{Hardness}
The complexity of a seeding policy is measured by the computability of $\pi(U,t,k)$. When solving the TAIM problem, we essentially consider two questions: (a) how many seed nodes to select and (b) which nodes to select. We refer the solution to the first question as a \textit{seeding pattern}. Not very surprisingly, both questions are computationally hard, and thus efficient optimal solutions are pessimistic. First, while TAIM does not generalize IM, we can use a reduction similar to that in \cite{kempe2003maximizing}. In particular, when the underlying graph is directed and bipartite with $p_e=1$ for each edge $e$, the TAIM problem generalizes the maximum coverage problem in a straightforward manner. 
Second, there exists an instance of TAIM of which the hardness is resulted from designing optimal seeding patterns but not from selecting seed nodes, which indicates that TAIM is combinatorially different from IM and AIM. 

\begin{lemma}
	\label{lemma: hardness_2}
	Even if the optimal seed nodes can be computed in polynomial time, the optimal policy for TAIM is not polynomial-time computable unless the decision version of s-t connectedness can be solved in polynomial time.
\end{lemma}
\begin{proof}
See Appendix.
\end{proof}

\subsection{Adaptive Gap}
For an instance $\I$ of TAIM, let $A_{opt}^{\I}\define \max_{\pi}f(\pi,T, K)$ be the maximum influence resulted by any policy, and \[N_{opt}^{\I}\define \max_{S \subseteq V, |S|\leq K}f(S,T, K)\] be the maximum influence resulted by a non-adaptive policy. The adaptive gap is defined as  $\sup \frac{A_{opt}^{\I}}{N_{opt}^{\I}}$ over the instances of TAIM, which measures the worst-case performance of the optimal non-adaptive policy compared to the optimal adaptive policy. Since non-adaptive policies can often be efficiently computed, one can adopt a non-adaptive one if the adaptive gap is small. In this paper, we provide a lower bound of the adaptive gap.

\begin{lemma}
	\label{lemma: lower_bound}
	The adaptive gap for the TAIM problem is at least $\frac{e^2-2}{e-1}$.
\end{lemma}
\begin{proof}
	The proof is inspired by the analysis in \cite{chen2019adaptivity} for the unconstrained case, while our problem involves a time constraint. For a certain integer $N \in \neqZ$, let us consider a directed line with $2N+1$ nodes with edges $(v_i,v_{i+1})$, $i \in \{1,...,2N\}$, where each edge has the probability $p=1-1/N$. Suppose that the time constraint is $2N$ and $K=2$. For each node $s_i$ with $i\leq N$, the expected influence resulted from $v_i$ within $t$ diffusion rounds is \[S(t)=\sum_{i=1}^{t}p^{i-1}\cdot (1-p)\cdot i+p^{t} \cdot (t+1)=\frac{1-p^t}{1-p}+p^t.\]
	
	Let us first consider an adaptive policy that (a) selects $v_1$ as the first node, (b) waits for the diffusion process terminates or the time limit is reached, and (c) selects the inactive node that is closest to $v_1$ as the second the seed node. The resulted influence would be
	
	\begin{align*}
	\Delta_{ad} &\define\sum_{i=1}^{2N-1}p^{i-1}(1-p)(i+S(2N-i))+p^{2N-1}(2N+1)\\
	&=2N(1-(1-\frac{1}{N})^{2N-1})-(2N-1)(1-\frac{1}{N})^{2N}\\
	&\hspace{4cm}+2(1-\frac{1}{N})^{2N-1}.
	\end{align*}
	
	For the non-adaptive case, the probability that a node $v_i$ can be activated is determined by the distance to it from the closest seed node. Therefore, the optimal non-adaptive influence should select $v_1$ and $v_{N}$ as the seed nodes, which follows from the fact that, supposing that another two nodes $v_{j_1}$ and $v_{j_2}$ were selected with $j_1 < j_2$, we could have a higher influence by first replacing $v_{j_1}$ by $v_1$ and then replacing $v_{j_2}$ by the mid node $v_N$.\footnote{Note that selecting $v_1$ and $v_{N+1}$ is another optimal non-adaptive policy.} Therefore, the optimal influence under a non-adaptive policy will be
	\begin{align*}
	\Delta_{non-ad} &\define \sum_{i=0}^{N-1}p^i+\sum_{i=0}^{N}p^i\\
	&=2N(1-(1-1/N)^N)+(1-1/N)^N.
	\end{align*}
	Now we have that $\frac{A_{opt}^{\I}}{N_{opt}^{\I}}$ is no less than $\frac{\Delta_{ad}}{\Delta_{non-ad}}$ of which the limit is $\frac{e^2-2}{e-1} \approx 3.14$.
\end{proof}

\section{Seeding Policy Design}
\label{sec: strategy}
In this section, we present several seeding policies. We first discuss the node selection rule and then design seeding policies based on different seeding patterns. Given the hardness in Sec. \ref{sec: theory}, we aim at the solutions that are (a) approximation solutions to the special cases in Remark \ref{remark: relate} and (b) effective heuristics for the general cases.

\subsection{Node Selection Rule}
\label{subsec: nodeselection}
When a seed set of a given size $k \in \mathbb{Z}^+$ is planned to be selected, the greedy rule is the most common method used in the existing studies. Supposing that $U$ is the current status,  we use
$g(U,S,t)$ to denote the expected number of active nodes after $t \in  \neqZ $ rounds following $U$ with $S \subseteq V$ being selected as the seed set. The local optimal solution would be 
\begin{equation}
\label{eq: greedy}
\argmax_{|S|=k} g(U,S,t).
\end{equation}
Computing the above equation is $NP$-hard, but the greedy rule, as shown in Alg. \ref{alg: greedy}, gives a $(1-1/e)$-approximation. Due to the $\#P$-hardness in computing $g(U,S,t)$, such a greedy rule is often implemented through stochastic optimization in which the key ingredients are (a) an unbiased estimator of the objective function $g(U,S,t)$ and (b) an estimate of the lower bound of $\argmax_{|S|=k} g(U,S,t)$. In the rest of this part, we will show how to obtain these ingredients.

An unbiased estimator of $g(U,S,t)$ can be obtained by the samples generated in the following sampling process. 
\begin{definition}[\textbf{RR-set}]
	\label{def: rr-set}
	Given the status $U$ and the remaining diffusion rounds $t$, an RR-set $\R$ is generated by:
	\begin{itemize}
		\item Step 1: select a node $v$ from $V$ uniformly at random.
		\item Step 2: simulate the diffusion process from $v$ in a \textit{reverse} direction in the manner of BFS. The simulation process terminates if (a) any node in $\dA(U)$ is encountered, (b) no node can be further reached, or (c) $t$ rounds of BFS has been executed. 
		\item Step 3. If the simulation terminates under the case (a) in Step 2, return $\R=V$ as the output. Otherwise, let $\R \subseteq V$ be the set of the nodes traversed during the simulation process, and return $\R$.
	\end{itemize}
\end{definition}
For each $V_1, V_2 \subseteq V$, we use 
\[
\I(V_1 \cap V_2)\define
\begin{cases}
1 &  \hspace{0mm} \hspace{-0.5mm} \text{if $V_1 \cap V_2 \neq \emptyset$} \\
0 & \hspace{0mm} \hspace{-0.5mm} \text{else } 
\end{cases}
\]
to denote if their intersection is empty, and for each $S \subseteq V$, let us consider the random variable $\I(S \cap \R)$. It turns out $n\cdot \I(S\cap \R)$ is an unbiased estimate of $g(U,S,t)$. 
\begin{lemma}
	\label{lemma: rrset}
	For each status $U$, $t \in \mathbb{Z}^+$, and $S \subseteq V$, we have 
	\begin{equation}
	\label{eq: lemma_rreset}
	n \cdot \E [\I(S\cap \R)]=g(U,S,t)
	\end{equation}
	where the expectation is taken over all possible $\R$ or equivalently over all the states of the edges.
\end{lemma}
\begin{proof}
	For each $v \in V$, let $g_v(U,S,t)$ be the probability that $v$ can be active after $t$ rounds following $U$ when $S$ is selected as the seed set, and thus we have $g(U,S,t)=\sum_{v \in V}g_v(U,S,t)$ due to the linearity of expectation. On the other hand, let $\R_v$ be the random RR-set when $v$ is selected in Step 1 in Def. \ref{def: rr-set}. By the sampling process, we have $\E [\I(S \cap \R)]=\sum_{v \in V}\frac{\E [\I(S \cap \R_v)]}{n}$ for each $v$. Therefore, it suffices to prove $\E [\I(S \cap \R_v)]=g_v(U,S,t)$. Since the expectation is taken over all the possible states of the edges, it further suffices to show that $\I(S \cap \R_v)=g_v(U,S,t)$ in each possible outcome of the edge states.\footnote{There are totally $2^{|E|}$ possible outcomes.} Now suppose that the states of the edges are fixed, and consequently $\I(S \cap \R_v)$ and $g_v(U,S,t)$ are binary-valued. According to the diffusion process of IC model, $g_v(U,S,t)=1$ iff there is a path of at most $t$ live edges  from $S \cup \dA(U)$ to $v$. According to the sampling process in Def. \ref{def: rr-set}, $\I(S \cap \R_v)=1$ iff $\dA(U) \cap \R_v \neq \emptyset$ (case (a) in Step 2) or $S \cap \R_v \neq \emptyset$ (case (b) and (c) in Step 2), which is $(\dA(U)\cup S) \cap \R_v \neq \emptyset$. Since $\R_v$ contains exactly the nodes that have a path of at most $t$ live edges to $v$, we have $g_v(U,S,t)=1$ iff $\I(S \cap \R_v)=1$.
\end{proof}

\begin{algorithm}[t]
	\caption{SOF Policy}\label{alg: t-rr}
	\begin{algorithmic}[1]
		\State \textbf{Input} {the current status $U$ and the remaining diffusion rounds $t$}
		\State Step 1: Uniformly select a random node that are node active in $U$;
		\State $k^*= \argmax_{i} \overline{\beta}(U,k,t, i)$;
		\State Return $\gr(U,k^*,t)$;
	\end{algorithmic}
\end{algorithm}

Suppose that a collection $\R_l$ of $l$ RR-sets were randomly generated, Lemma \ref{lemma: rrset} immediately implies that when $l$ is sufficiently large, the $S$ that can maximize 
\begin{equation}
\label{eq: estimate}
\frac{\sum_{\R \in \R_l} n \cdot \I(S\cap \R)}{l}
\end{equation}
should be able to maximize $g(U,S,t)$.  We can easily see that Eq. (\ref{eq: estimate}) is submodular with respect to $S$ for each $\R_l$, and therefore, greedy algorithm gives a $(1-1/e)$-approximation to $\argmax_{|S|=k}\frac{\sum_{\R \in \R_l} n \cdot \I(S\cap \R)}{l}$.  Using the central inequality (e.g. Chernoff bound) to bound the estimation accuracy requires the lower bound of $\argmax_{|S|=k} g(U,S,t)$, which is however not known to us in advance. Fortunately, because $g(U,S,t)$ is bounded within $[|\dA(U)|,n]$ and its estimate Eq. (\ref{eq: estimate}) can be approximated within a factor of $1-1/e$, we could utilize the adaptive sampling method designed in \cite{tang2014influence} to search a lower bound that is within a constant factor to $\argmax_{|S|=k} g(U,S,t)$. Following the standard reverse sampling analysis, we have the following result.
\begin{lemma}[\hspace{1sp}\cite{tang2014influence}]
	\label{lemma: greedy}
	There exists a greedy algorithm that can produce a $(1-1/e-\epsilon)$-approximation to Eq. (\ref{eq: greedy}) with probability at least $1-n^{-l}$ within time $O(\frac{(k+l)(m+n)\log n}{\epsilon^2})$, for each $\epsilon \in (0,1)$ and $l\geq 1$. 
\end{lemma}
\begin{proof}
	The proof follows directly from the analysis in Sec. 3 of \cite{tang2014influence} with the only difference that the new RR-set defined in Def. \ref{def: rr-set} is used.
\end{proof}
We use $\gr(U,t,k)=\{v_1,...,v_k\}$ to denote the output of greedy algorithm in Lemma \ref{lemma: greedy}, and we assume that the indexes follow the order in which the nodes were selected (e.g., $v_1$ was the first node added to the solution). In this paper, we will utilize this algorithm for node selection, enabling us to focus primarily on designing seeding patterns. One plausible reason for doing so is that, for the special cases discussed in Remark \ref{remark: relate}, the greedy node selection rule is the best polynomial method in terms of the approximation ratio.  

For the convenience of discussion, we introduce the following definitions that will be used in the rest of this section.

\begin{definition}[\textbf{Status} $U_{t,k}$]
	\label{def: utk}
	For a status $U$ and two integers $t,k \in \neqZ$, we use $U_{t,k}=\big(\dA(U)\cup \gr(U,t, k), \dL(U), \dD(U)\big)$ to denote the status when $\gr(U,t, k)$ is selected as the seed set.
\end{definition}

\begin{definition}[\textbf{Future Status} $\U_t(U)$]
	\label{def: utu}
	For a status $U$,  we use $\U_t(U)$ to denote the set of all possible status after $t \in \neqZ$ diffusion rounds following $U$ without selecting any new seed node, and we use $\Dis({\U}_t(U))$ to denote the associated distribution over $\U_t(U)$.
\end{definition}

\begin{algorithm}[t]
	\caption{Greedy Node Selection Rule$(U, t, k)$}
	\label{alg: greedy}
	\begin{algorithmic}[1]
		\State \textbf{Input:} $(U, t, k)$;
		\State \textbf{Output:} a user set $S$;
		\State $S \leftarrow \emptyset$;
		\For {${i=1:k}$}
		\State $v^* \leftarrow \argmax_v g(U, S+v, t)$;
		\State $S \leftarrow S+v^*$;
		\EndFor
		\State Return $S$;
	\end{algorithmic}
\end{algorithm}

\subsection{Seeding Policy}

If we would always use the greedy rule for node selection, the problem left is to decide the seeding pattern: given a status, how many seed nodes should be selected in each seeding step? As aforementioned, in each seeding step, we would select seed nodes as few as possible to maximally utilize the merits of adaptive seeding, while we would select seed nodes as many as possible to have more future diffusion rounds by the time constraint. With such intuitions in mind, we propose five seeding policies.

\subsubsection{{Basic Seeding Policy}} In general, a seeding pattern can be either specified before the seeding process, or dynamically constructed during the seeding process. For the TAIM problem, an immediate solution is to utilize a \textit{static seeding pattern}, which is given by a sequence $(a_1,..,a_T)$ of non-negative integers where $a_i \in \neqZ$ is the size of the seed set in the $i$-th seeding step. Under the budget constraint, we have $\sum a_i \leq K$. Combining with the greedy node selection rule, we have a basic policy:
\begin{policy} [\textbf{Static Policy} $(a_1,...,a_T)$]
	\label{policy: static}
	In the $i$-th seeding step with state $(U, t, k)$, select $\gr(U,t,a_i)$ as the seed set.
\end{policy} 

A static policy is relatively simple to implement, but one drawback is that we have less knowledge on finding the sequence $(a_1,..,a_T)$ so that the influence can be maximized. Note that the searching space is exponential, and the hardness in finding the optimal static pattern can be additionally seen from the proof of Lemma \ref{lemma: hardness_2}. As a preliminary solution, we propose the \textit{$k$-filter uniform pattern} where the seeding actions are uniformly distributed to the diffusion period, and each seeding step selects the same number of seed nodes. Formally, given the filter size $k \in \neqZ$, we aim to achieve the pattern where $a_{1}=a_{1+k}=...=a_{1+(d-1)*k}=\lfloor K/d \rfloor$ with $d=\lfloor T/k \rfloor$. For example, under $T=10$, $K=50$ and $k=2$, we have the pattern $(10,0,10,0,10,0,10,0,10 ,0)$. In the case that we have the remaining budget due to the rounding, we will use them right before the last diffusion round. 


The static policy cannot fully utilize the merits of the adaptive policy as the seeding patterns are fixed. In the next, we present the \textit{greedy seeding pattern} in which we keep observing the diffusion process until a final status is reached or no diffusion round is left. More specifically, when there are more than one diffusion rounds left, we select one seed node if the current status is final, and otherwise wait for more results; when there is only one diffusion round left, we immediately use up the remaining budget. 

\begin{policy}[\textbf{Greedy Policy}]
	\label{policy: greedy}
	In each seeding step with state $(U, t, k)$:
	\begin{itemize}
		\item If $t=1$, select $\gr(U,t,k)$ as the seed set.
		\item If $t>1$ and $U$ is final, select $\gr(U,t,1)$  as the seed set.
		\item If $t>1$ and $U$ is not final, no seeding action is performed.
	\end{itemize}
\end{policy} 

Compared to the static seeding pattern, the greedy policy does not require any input prior to the seeing process, and it always attempts to obtain the diffusion results maximally. However, because it does not consider the time constraint until the last diffusion round, it might be the case that the majority of the seed nodes are used right before the last diffusion round and thus have very limited time to spread widely, which can be observed later in the experiments in Sec. \ref{sec: exp}. 

\subsubsection{{One-step Foresight Seeding Policy}} Further generalizing the greedy seeding pattern, we propose the \textit{one-step foresight seeding pattern} in which we will estimate the profit of selecting a particular number of seed nodes in each step. Given the state $(U,k,t)$, we consider the scenario that $k_1\leq k$ nodes are first selected in the current step and $k_2=k-k_1$ nodes will be selected after one round of diffusion. When $k_1$ nodes are selected in the current seeding step by the greedy node selection rule, the optimal profit, denoted as $\beta(U,k, t,k_1)$, will be 

\begin{align}
\label{eq: beta}
&\beta(U,k, t,k_1) \nonumber\\
&\define \sum_{U^* \in {\U}_1(U_{t, k_1})}\Pr[U^*|U_{t, k_1}] \cdot  \max_{|S|=k-k_1} g(U^*,S,t-1)\nonumber \\
&=\E_{U^* \sim \Dis({\U}_1(U_{t, k_1}))} \Big[\max_{|S|=k-k_1} g(U^*,S,t-1)\Big]
\end{align}
where $\Pr[U^*|U_{t, k_1}]$ is the probability that $U^*$ happens conditioned on $U_{t, k_1}$, and $U_{t, k_1}$ and  ${\U}_1(U_{t, k_1})$ are given by Defs. \ref{def: utk} and \ref{def: utu}. Consequently, the optimal $k_1$ under this pattern is 
\begin{equation}
\label{eq: time_size}
K(U,t, k)\define\argmax_{k_1\leq k} \beta(U,k, t,k_1),
\end{equation}
based on which we design the One-step Foresight (\textbf{OF}) policy:

\begin{policy}[\textbf{One-step Foresight (OF) Policy}]
	In each seeding step with state $(U, t, k)$, select $\gr(U,t,k^*)$ as the seed set where $k^*=K(U,t,k)$.
\end{policy}

While the OF policy only considers the one-step forward foresight, it indeed provides the best possible solution for all the special cases mentioned in Remark \ref{remark: relate}. 
\begin{lemma}
	\label{lemma: OF}
	The OF policy provides a $(1-1/e)$-approximation if either (a) $T=1$, (b) $T=\infty$ or (c) $p_e=1$ for each edge $e \in E$.
\end{lemma}
\begin{proof}
	First, it produces a $(1-1/e)$-approximation when $T=1$. This is because (a) $K(U,t, k)=k$ when $t=1$ and (b) the nodes are selected by the greedy node selection rule. 
	
	Second, it is a $(1-1/e)$-approximation when $t=\infty$. First, at a certain seeding step when the status $U$ is not final, we always have $\beta(U,k,\infty,i)\leq \beta(U,k,\infty,0)$ for each $i\leq k$, which follows from the fact that waiting for the diffusion to complete is always optimal if there is no time constraint. Therefore, the OF policy will always wait for the diffusion to reach a final status before selecting the next seed set. Second, at a certain seeding-step when the status $U$ is final, we have $\beta(U,k,\infty,i)\leq \beta(U,k,\infty,1)$ for each $i \leq k$, and therefore we have $K(U, k, \infty)=1$, which means we always wait for the diffusion to complete and always select one seed node whenever a seed set should be selected. As a result, OF follows the full-adoption feedback model and selects each seed node in a greedy manner, which yields a $(1-1/e)$-approximation. 
	
	Finally, under the deterministic independent cascade model, any adaptive policy is in fact non-adaptive as the diffusion process has no uncertainty, so we must have $\beta(U,k,t,k)\geq \beta(U,k,t,i)$ for each $i \leq k$. Combining the greedy node selection rule, OF will again give a $(1-1/e)$-approximation.
\end{proof}
One can see that either the static policy (Policy \ref{policy: static}) or the greedy policy (Policy \ref{policy: greedy}) cannot always guarantee the same for those special cases. While the OF policy can take account of the time constraint in each seeding step, it is not practically feasible in terms of the computability because $\beta(U,k, t,k_1)$ is hard to compute due to the fact that (a) computing $\max_{|S|=k-k_1} g(U^*,S,t-1)$ is $NP$-hard and (b) there can be an exponential number of terms in ${\U}_1(U_{t, k_1})$. To deal with the first issue, we use the $(1-1/e)$-approximation as an estimate of $\argmax_{|S|=k-k_1} g(U^*,S,t-1)$, and therefore, the quantity we are interested in is 
\begin{align*}
\label{eq: over_beta}
&\overline{\beta}(U,k,t,k_1) \define\\
&\E_{U^* \sim \Dis({\U}_1(U_{t, k_1}))} \Big[g(U^*,\gr(U^*,t-1, k-k_1),t-1)\Big].
\end{align*}
For the second issue, we can estimate $\overline{\beta}(U,k,t,k_1)$ through sampling. In particular, given the input $(U,t, k)$ and $k_1$, the estimation can be obtained through samples generated by the following procedure:
\begin{enumerate}
	\item obtain $U_{t,k_1}$ using $\gr(U,t,k_1)$,
	\item sample a status $U^*$ following $\Dis({\U}_1(U_{t, k_1}))$ by simulating the diffusion process for one round, and,
	\item compute $g(U^*,\gr(U^*,t-1,k-k_1),t-1)$ as an estimate of $\overline{\beta}(U, k,t,k_1)$.  
\end{enumerate}
Supposing $L$ simulations are used for each estimation, the resulted policy is shown in Alg \ref{alg: OF policy}, denoted as the Sampling-enhanced One-step Foresight (\textbf{SOF}) policy.

\begin{policy}[\textbf{Sampling-enhanced One-step Foresight (SOF) Policy} $(L \in \neqZ)$]
	\label{policy: SFOM}
	In each seeding step with state $(U, t, k)$, select the seed set obtained by running Alg. \ref{alg: OF policy} with input $(U, t, k)$ and $L$.
\end{policy} 

\begin{algorithm}[t]
	\caption{SOF Policy}\label{alg: OF policy}
	\begin{algorithmic}[1]
		\State \textbf{Input} {$(U, k, t)$ and $L$}
		\For {$i=0:k$}
		\State Estimate $\overline{\beta}(U,k,t, i)$ by $L$ simulations;
		\EndFor
		\State $k^*= \argmax_{i} \overline{\beta}(U,k,t, i)$;
		\State Return $\gr(U,k^*,t)$;
	\end{algorithmic}
\end{algorithm}

Comparing to the static policy and greedy policy, the SOF policy is more sophisticated, but it incurs a higher complexity because we have to invoke the greedy node selection rule in each sampling. In our experiments, we have observed that the SOF policy is not scalable to handle large datasets. 

\subsubsection{Fast Foresight Seeding Policy}
For supporting high-volume datasets, we finally present a fast foresight seeding policy. Our design is driven by considering a fine-grained trade-off between seeding and observing. Given the state $(U,t,k)$ with $\gr (U,t,k)=\{v_1,...,v_k\}$ being the local greedy solution and $S_i\define \{v_1,...,v_i\}$, our method considers the node one by one from $v_1$ to $v_k$ and determines if they would be selected in the current solution. This framework is formally described in Alg. \ref{alg: framework} where $\ifadd(U,S_i,v_{i+1},t) \in \{\true,\false\}$  is a module for determining if $v_{i+1}$ should be selected given that $S_i$ has already been selected. Once a node $v_{i+1}$ is rejected by $\ifadd()$, the process terminates and takes $S_i$ as the seed set in the current seeding step. Such a framework enables us to concentrate on analyzing the marginal effect of adding an individual node - designing the $\ifadd()$ module. To this ends, we propose a novel and efficient $\ifadd()$ designed through two quantities, $\Ma(U,S_i,v_{i+1}, t) \in [0,1]$ and $\Mt(U,S_i,v_{i+1}, t) \in [0,1]$, which reveals the gain or loss of selecting or not selecting $v_{i+1}$. In particular, when $\Ma(U,S_i,v_{i+1}, t)$ or $\Mt(U,S_i,v_{i+1}, t)$ approaches $1$, it is a strong indicator for selecting $v_{i+1}$ as another seed node. When $\Ma(U,S_i,v_{i+1}, t)$ or $\Mt(U,S_i,v_{i+1}, t)$ is close to $0$, it is a strong indicator for not selecting $v_{i+1}$. In what follows, we present these two metrics in detail.

\begin{algorithm}[t]
	\caption{Fast Seeding Policy Framework}\label{alg: framework}
	\begin{algorithmic}[1]
		\State \textbf{Input} {$(U, k,t)$}
		\State $\{v_1,...,v_k\} \leftarrow \gr(U,k,t)$;
		\State $S_i \leftarrow \{v_1,...,v_i\}$;
		\For {$i=0:k$}
		\If {$\ifadd(U, S_{i}, v_{i+1}, t) == \false$}
		\State Return $S_i$;
		\EndIf
		\EndFor
		\State Return $S_k$;
	\end{algorithmic}
\end{algorithm}

The first metric $\Ma()$ is designed by measuring the correlation between the influence triggered by different seed nodes. Due to the diminishing marginal return of influence, the marginal contribution of a node will decrease after other nodes have been selected, which is an essential reason that a higher influence can be achieved by allowing an adaptive seeding. However, if the influences resulted from different seed nodes were independent, observing the feedback is not useful, and consequently an adaptive policy only incurs the loss of diffusion rounds. For example, if $v_1$ and $v_2$ are in different connected components of the graph, observing the influence resulted from $v_1$ does not alter the capability of $v_2$ in terms of influencing other users. From this perspective, the next lemma gives a sufficient condition for testing such independence.

\begin{lemma}
	\label{lemma: sufficient}
	For a seeding step with state $(U,t,k)$ and a seed set $S^* \subseteq V$, observing the cascade resulted by any subset $S_1 \subseteq S^*$ will not alter the marginal contribution of $S^* \setminus S_1$, provided that 
	\begin{equation}
	\label{eq: sufficient}
	g(U, S^*, t)-g(U,\emptyset, t)=\sum_{v\in S^*}(g(U, v, t)-g(U, \emptyset, t)).
	\end{equation}
\end{lemma}
\begin{proof}
Let us use $\E_{G_l \sim\G_U}[g_u(U, S^*, t|G_l)]$ to denote the probability that a node $u$ can be activated after $t$ rounds following $U$ when $S^*$ is selected, where $G_l$ is a sampled live-edge graph conditioned on $U$ and  $g_u(U, S^*, t|G_l) \in \{0,1\}$ is an indicator function denoting if there exists a path in $G_l$ from $A(U)\cup S^*$ to $v$ with no more than $t$ live edges. Due to the linearity of expectation, we have $g(U, S^*, t)=\sum_u \E_{G_l\sim\G_U}[g_u(U, S^*, t|G_l)]$. Because $g_u(U, S^*, t|G_l)$ is submodular, $g_u(U, S^*, t|G_l)-g_u(U, \emptyset, t|G_l)$ is no larger than $\sum_{v\in S^*}(g_u(U, v, t|G_l)-g_u(U, \emptyset, t|G_l))$.
Combining Eq. (\ref{eq: sufficient}), it implies that 
\begin{align}
\label{eq: equal}
g_u(U, S^*, t|G_l)&-g_u(U, \emptyset, t|G_l) \nonumber \\
&= \sum_{v\in S^*}(g_u(U, v, t|G_l)-g_u(U, \emptyset, t|G_l))
\end{align}
holds for each $G_l$ and $u$. Therefore, whenever the LHS of Eq. (\ref{eq: equal}) is equal to 1, there must be exactly one term on the RHS of Eq. (\ref{eq: equal}) is equal to 1, Taking $G_l$ as the live graph where all the edges in $E \setminus (\dD(U)\cup \dL(U))$ are live, it implies that any inactive node cannot be connected to two nodes in $S^*$ through a path in $E \setminus (\dD(U)\cup \dL(U))$ with edges less or equal than $t$. As a result, the cascades triggered by the nodes in $S^*$ are totally independent, and therefore observing the cascade resulted by any subset $S_1 \subseteq S^*$ does not change the marginal gain of any node in $S^* \setminus S_1$.
\end{proof}

The intuition behind Lemma \ref{lemma: sufficient} is that the overlapping of the contributions can be evaluated by testing the submodularity. For our problem, to determine if $v_{i+1}$ will be selected given $(U_{t,i},t,k)$, Lemma \ref{lemma: sufficient} suggests that we can do this by comparing $g(U, S_i+v_{i+1}, t)-g(U,S_i, t)$ and $g(U,v_{i+1}, t)-g(U,\emptyset, t)$. As a result, we can leverage the quantity \[\Ma(U,S_i, v_{i+1}, t)\define \dfrac{g(U, S_i+v_{i+1}, t)-g(U,S_i,  t)}{g(U,v_{i+1}, t)-g(U,\emptyset, t)},\] 
which measures the benefits of including $v_{i+1}$ in the seed set of the current seeding step. When we have $\Ma(U,S_i, v_{i+1}, t)=1$, there is a good reason to include $v_{i+1}$ in the current step because observing more diffusion results will not decrease the marginal gain of $v_{i+1}$. When $\Ma(U,S_i, v_{i+1}, t)$ is close to $0$, it simply means that the marginal gain of $v_{i+1}$ is small, and therefore we may wait for more observations to have better seed nodes.

The second metric $\Mt()$ is designed by looking into the one-step loss of the influence resulted by $v_{i+1}$. Given $(U,t,k)$, if the total influence resulted by $v_{i+1}$ can always complete within $t-1$  diffusion rounds, we have a good reason for not selecting it in the current seeding step in that there is no loss of waiting for another diffusion round. Formally, consider the status $U_{t,i}=(A(U)\cup S_i,\dL(U), \dD(U))$ with $\U_t(U_{t,i})$ being its future status after $t$ diffusion rounds. When the influence from $v_{i+1}$ is allowed to spread for $t^* \in \neqZ$ diffusion rounds, the marginal gain of selecting $v_{i+1}$ would be 
{ 
	\begin{align*}
	&h(U,S_{i},v_{i+1},t,t^*)\\
	&\define \E_{U^* \sim \Dis(\U_t(U_{t,i}))}\Big[g(U^*,v_{i+1},t^*)-|A(U^*)|\Big].
	\end{align*}
}
The above formula can be explained as: we first simulate the influence from $S$ for $t$ rounds to obtain a status $U^* \in \U_t(U_{t,i})$, and conditioning on $U^*$ we then simulate the influence from $v_{i+1}$ for $t^*$ rounds. Viewing the diffusion process from a multi-step perspective shares the same insights in \cite{mossel2007submodularity}. Now let us utilize the quantity
\begin{align*}
&\Mt(U,S_{i},v_{i+1},t)\\
&\define \frac{h(U,S_{i},v_{i+1},t,t)-h(U,S_{i},v_{i+1},t,t-1)}{h(U,S_{i},v_{i+1},t,t)}
\end{align*}
to measure the loss incurred by seeding $v_{i+1}$ with a delay of one round. If $\Mt(U,S_{i},v_{i+1},t)=0$, there would be no such a loss and we thus should not select $v_{i+1}$ in the current seeding step. One the other hand, if we have $\Mt(U,S_i,v_{i+1},t)$ close to $1$ (i.e., $h(U,S_{i},v_{i+1},t,t-1)$ is small), it means that $v_{i+1}$ can hardly trigger any influence if seeded one round later, and therefore, we prefer to select it immediately. For example, we always have $\Mt(U,S_i,v_{i+1},t)=1$ when there is only one remaining diffusion round (i.e., $t=1$).

With the metrics $\Ma()$ and $\Mt()$, let us consider the quantity
\begin{align*}
&\Ind(U,S_i,v_{i+1},t)\\
&\define\alpha(t)\cdot \Ma(U,S_i, v_{i+1}, t)+\big(1-\alpha(t)\big)\cdot\Mt(U,S_i, v_{i+1}, t),
\end{align*}
where $\alpha(t) \define 1-1/t \in [0,1]$.  According the above design, given the state $(U,t,k)$, we would select $v_{i+1}$ when $\Ind(U,S_i,v_{i+1},t)$ approaches to $1$, while not select $v_{i+1}$ when it approaches to $0$. We can see that $\alpha(t)$ is a balancing parameter, and it becomes larger when fewer diffusion rounds remain, which increases the importance of $\Mt()$ when the time constraint is severe. As a result, the module $\ifadd()$ can be constructed as
\begin{equation}
\label{eq: ind}
\ifadd(U,S_i, v_{i+1}, t) \define
\begin{cases}
\true &  \hspace{0mm} \hspace{-0.5mm} \text{if $\Ind(U,S_i, v_{i+1}, t) \geq \theta$} \\
\false & \hspace{0mm} \hspace{-0.5mm} \text{otherwise } 
\end{cases}
\end{equation} 
where $\theta \in [0,1]$ is a controllable threshold that can reflect certain prior knowledge or preference. For instance, adopting a small $\theta$ could result in more seed nodes in the first several diffusion rounds. The effect of $\theta$ will be further investigated through experiments. We denote the resulted policy as the Fast Foresight (\textbf{FF}) policy:

\begin{policy}[\textbf{Fast Foresight (FF) Policy $(\theta \in (0,1))$}]
	\label{policy: FOFM}
	In each seeding step with state $(U,t,k)$, the seed set is computed by Alg. \ref{alg: framework} with $\ifadd()$ given by Eq. (\ref{eq: ind}).
\end{policy} 

In addition to the subroutine $\gr()$, implementing the FF policy requires to compute $\Ma()$ and $\Mt()$, which can be estimated again by sampling. Although both the SOF policy and FF policy involve the sampling procedure, it can be shown both theoretically and experimentally that FF is more efficient than SOF. 


\begin{table}[t]
	\caption{Time Complexity.}
	\centering
	\label{table: time}
	\begin{tabular}{@{}c|l@{}}
		\toprule
		Static Policy & $O\big((k+l) (m+n)\cdot \log n/\epsilon^2\big)$  \\
		Greedy Policy & $O\big((k+l)(m+n)\cdot \log n/\epsilon^2\big)$ \\
		SOF Policy & $O\big((k^2+kl) (m+n)\cdot \log n \cdot L/\epsilon^2\big)$  \\
		FF Policy & $O\big((k+l) (m+n)\cdot \log n/\epsilon^2+Lk^2\cdot (m+n)\big)$ \\
		\bottomrule
	\end{tabular}
\end{table}

\begin{lemma}
	\label{lemma: complexity}
	Suppose that we use $L \in \neqZ$ samples for each estimation in SOF and FF, and the parameters used in $\gr()$ are $\epsilon$ and $l$. The complexity of the policies is given by Table \ref{table: time}.
\end{lemma}
\begin{proof}
Recall that for an adaptive seeding policy, we measure its complexity by the running time of computing one seed set. The static policy invokes the greedy node selection rule in each step so its time complexity is $O((K+l) (m+n)\cdot \log n/\epsilon^2)$. The greedy policy first examines if the status is final, which can be done in $O(m+n)$, and therefore the total time complexity is again $O((K+l)(m+n)\cdot \log n/\epsilon^2)$. In the SOF policy, each estimation in line 3 in Alg. \ref{alg: OF policy} consists of two parts: the simulation, running in $O(m+n)$, and the greedy node selection rule, running in $O((K+l)(m+n)\cdot \log n /\epsilon^2)$. Therefore,  Alg. \ref{alg: OF policy} runs in $O((K^2+K l) (m+n)\cdot \log n\cdot L/\epsilon^2)$. In FF policy, line 2 runs in $O((K+l)(m+n)\cdot \log n/\epsilon^2)$, and each estimation of $\Ma()$ and $\Mt()$ can be done in $KL(m+n)$. Therefore, the complexity of FF is $O((K+l)(m+n)\cdot \log n/\epsilon^2+K^2L\cdot(m+n))$.
\end{proof}

\begin{table}[t]
	\caption{Dataset.}
	\centering
	\label{table: data}
	\begin{tabular}{@{} c|l l l l @{}}
		\toprule
		& Power & Wiki & Reddit & Youtube \\
		\midrule
		Nodes &2,500& 8,300 & 124,960 & 1,157,900 \\
		Edges &26,449& 103,689 & 624,349 &  5,975,248\\
		\bottomrule
	\end{tabular}
\end{table}

\section{Experiments}
\label{sec: exp}
In this section, we report the results of the experiments done for studying the practical performance of the proposed policies, aiming at examining (a) their ability to achieve a large influence, (b) the running time, and (c) the robustness of the seeding pattern.

\textbf{Datasets.} We adopt four datasets: (a) \textit{Power}: a synthetic power-law graph \cite{cowendigg}, (b)  \textit{Wiki}: a Wikipedia voting network \cite{leskovec2010signed}, (c) \textit{Reddit}: a graph inferred from Reddit social networking platform, and (d) \textit{Youtube}: a social network extracted from Youtube.com \cite{yang2015defining}. Reddit is a new dataset created in this paper. We collected 1,000 threads, each of which has at least 1,500 replies, from the News subreddit in August 2019 and constructed a graph with the users who have participated at least two threads. A brief summary of the datasets is given in Table. \ref{table: data}.

\begin{table*}[t]
	\centering
	\caption{Influence Resulted by Different Policies. }
	\label{table: influence}
	\begin{tabular}{@{}c|c c c c c c|c|c c c|c|c@{}}
		\toprule
		& \multicolumn{6}{c|}{FF}& \multirow{1}{*}{SOF} & \multicolumn{3}{c|}{Static} & \multirow{1}{*}{Greedy} & \multirow{1}{*}{\makecell{NonAd}}\\
		\midrule
		& $\theta=0.01$ &   $\theta=0.2$ & $\theta=0.26$ & $\theta=0.4$  & $\theta=0.6$ & $\theta=0.7$  &  & $k=1$ & $k=2$ &$k=5$ & &\\
		Power &$963.2$   &$971.9$ & $\textbf{981.1}$ &$878.7$ &$729.3$ &$653.4$  & $\textbf{979}$  & $975.7$&  $\textbf{994.7}$ & $\textbf{994.4}$ & $507.4$ & $932.3$ \\
		\hline
		& $\theta=0.4$ & $\theta=0.5$ & $\theta=0.6$ & $\theta=0.7$ & $\theta=0.8$ &  $\theta=0.9$&  & $k=1$ & $k=2$ &$k=5$ & &\\
		Wiki    &$681.0$ &$686.4$ &$\textbf{694.8}$ &$665.2$ &$538.6$ &$239.8$ & $\textbf{694.5}$ & $687.6$ & $687.2$ & $688.6$ & $493.1$ & $669.1$  \\
		\hline
		& $\theta=0.2$ & $\theta=0.3$ & $\theta=0.4$ & $\theta=0.5$ & $\theta=0.6$ & $\theta=0.7$ &    & $k=1$ & $k=2$ &$k=5$ & &\\
		Reddit  &$749.8$ &$761.0$ &$759.7$ &$\textbf{772.2}$ &$684.6$ & $471.3$  &  $\textbf{810.5}$ & $657.9$ & $674.7$ & $748.3$ & $181.6$ & $734.4$ \\
		\midrule
		\multirow{2}{*}{\makecell{Youtube\\$(10, 20, 0.01)$}} & $\theta=1$E-$4$& $\theta=1$E-$3$& $\theta=5$E-$3$ & $\theta=0.01$& $\theta=0.1$ & $\theta=0.2$ &   & $k=1$ &$ k=2$ &$k=5$ & &\\
		&$\textbf{7806}$ & $7791$& $7762$& $7762$& $7633$ & $7504$    &  n/a &$7647$ & $7747$& $\textbf{7873}$& $6800$& $7774$\\
		\hline
		\multirow{2}{*}{\makecell{Youtube\\$(10, 10, 0.005)$}} & $\theta=0.2$ & $\theta=0.3$& $\theta=0.4$ & $\theta=0.5$& $\theta=0.6$ & $\theta=0.7$  &  & $k=1$ & $k=2$ &$k=5$ & &\\
		& $898.3$&  $912.5$ & $911.2$ &$\textbf{955.2}$  &$\textbf{962.6}$   & $868.8$   &  n/a &$913.1$ &$903.1$ & $922.3$& $694.5$& $927.4$\\
		\hline
		\multirow{2}{*}{\makecell{Youtube\\$(10, 10, 0.001)$}} & $\theta=0.1$& $\theta=0.2$ & $\theta=0.5$& $\theta=0.6$ & $\theta=0.7$ & $\theta=0.8$    &  & $k=1$ & $k=2$ & $k=5$ & &\\
		& $53.6$ & $54.1$ & $53.9$ &$ 54.3$ & $56.0$ & $54.2$   &  n/a &$\textbf{82.3}$ & $\textbf{78.9}$ & $63.5$ & $71.8$& $54.6$\\
		\bottomrule
		\multicolumn{13}{r}{*Competitive results are in bold. }
		
	\end{tabular}
\end{table*}

\textbf{Settings.} For Power and Wiki, we consider the weighted-cascade setting where $p_{(u,v)}=1/\InDeg(v)$ with $\InDeg(v)$ being the in-degree of $v$. For Reddit, the probability on edge $p_{(u,v)}$ is proportional to the frequency between $u$ and $v$, where the frequency is measured by the number of the threads that both $u$ and $v$ have participated in. On Power, Wiki and Reddit, we consider the setting $(T,K)=(10,50)$. We adopt a short period in order to examine the ability of each policy to deal with a severe time constraint. For Youtube, each edge $e \in E$ has the same propagation probability $p_e$, and we adopt three settings:  $(T,K,p_e)=(10,20,0.01)$, $(10,10,0.005)$, and $(10,10,0.001)$. We will shortly see how these settings could help us investigate the property of the FF policy. The implementation of the greedy node selection rule follows the vanilla reverse sampling framework \cite{tang2015influence}.  For ensuring that each policy could run in a reasonable time, $500$ (resp., $50$) samples were used for each estimation in FF (resp., SOF). We tested the $k$-filter uniform pattern for the static policy with $k\in \{1,2,5\}$. We also tested the greedy policy and the non-adaptive policy, denoted as Greedy and NonAd, which can be taken as two baselines. For each dataset and each seeding policy, we repeated the experiment for 300 times and report the average result. Our experiments were done on an Intel Xeon Platinum 8000 Series processor with parallelizations. We wish to note that, to the best of our knowledge, there is no existing algorithm for the AIM problem that can meet a hard deadline. The analysis in \cite{vaswani2016adaptive} provided a general framework but no specific algorithm for AIM was studied in their experiments. Therefore, we focus on experimentally examining the performance of the policies proposed in this paper. 

\subsection{Results}
The resulted influence under each policy is given in Table \ref{table: influence}, and the running time is shown in Table \ref{table: running_time} in which the report of Youtube is for $(T,K,p_e)=(10,10,0.005)$. The results of SOF on Youtube is not reported because it was not able to complete within five days. Notice that all the methods are heuristic, and we did not have any prospect on either which policy would provide the best performance or the seeding patterns dynamically constructed during the seeding process. 

\textbf{Results on Power, Wiki, and Reddit.} According to Table \ref{table: influence}, the SOF policy can provide the most competitive performance, but it may take hours to compute one seed set, making it not suitable for time-sensitive tasks. Second, the FF policy is reasonably good, provided that an appropriate $\theta$ is used. On Reddit, the FF policy outperforms the static policy by an evident margin. The static policy can also produce moderate performance, and it gives the best result on Power with $k=2$. However, on Reddit, the static policy is worse than non-adaptive. Finally, the baseline methods, greedy policy and non-adaptive policy, are not effective. 

\textbf{Results on Youtube.} In order to test the extreme cases, we first consider Youtube with $(T, K, p_e)=(10,20,0.01)$. In such a case, we see that the non-adaptive policy has relatively good performance, while the greedy policy is very ineffective. This is because the resulted influence is very large, leading to that the time constraint dominates the merits of the adaptive seeding. Therefore, TAIM reduces to the TIM problem, and the need for adaptive seeding is low. In another extreme case, when we have $(T,K,p_e)=(10,10,0.001)$, the influence can hardly spread for more than one round due to the low prorogation probability, and therefore TAIM is close to AIM for which the Greedy policy is relatively good, which is supported by the results in Table \ref{table: influence}. In such a case, FF is again not effective due to the construction of $\ifadd()$. Note that such extreme cases are constructed artificially, and for the settings between those extreme cases, the FF policy can be effective, which can be seen from the results in Table \ref{table: influence} for the setting $(T,K,p_e)=(10, 10, 0.005)$. 

\begin{table}[t]
	\caption{Running Time. Each cell gives the average running time of computing one seed set.}
	\centering
	\label{table: running_time}
	\begin{tabular}{@{}c|l|l|l|l@{}}
		\toprule
		& Static & Greedy & SOF & FF\\
		\midrule
		Power  & $1.3$s & $<1$s & $21$min   &$6.0$s\\
		Wiki & $<1$s & $<1$s & $55$min & $2.5$s \\
		Reddit &$4.9$s & $4.9$s & $81.6$min & $50.1$s \\
		Youtube & $45.0$s & $44.0$s &  n/a  & $51.5$s\\
		\bottomrule
	\end{tabular}
	\vspace{-3mm}
\end{table}


\begin{figure*}[t]
	\centering
	\subfloat[{[Power, 0.2]} ]{\label{fig: power_02_pattern}
		\includegraphics[width=0.16\textwidth]{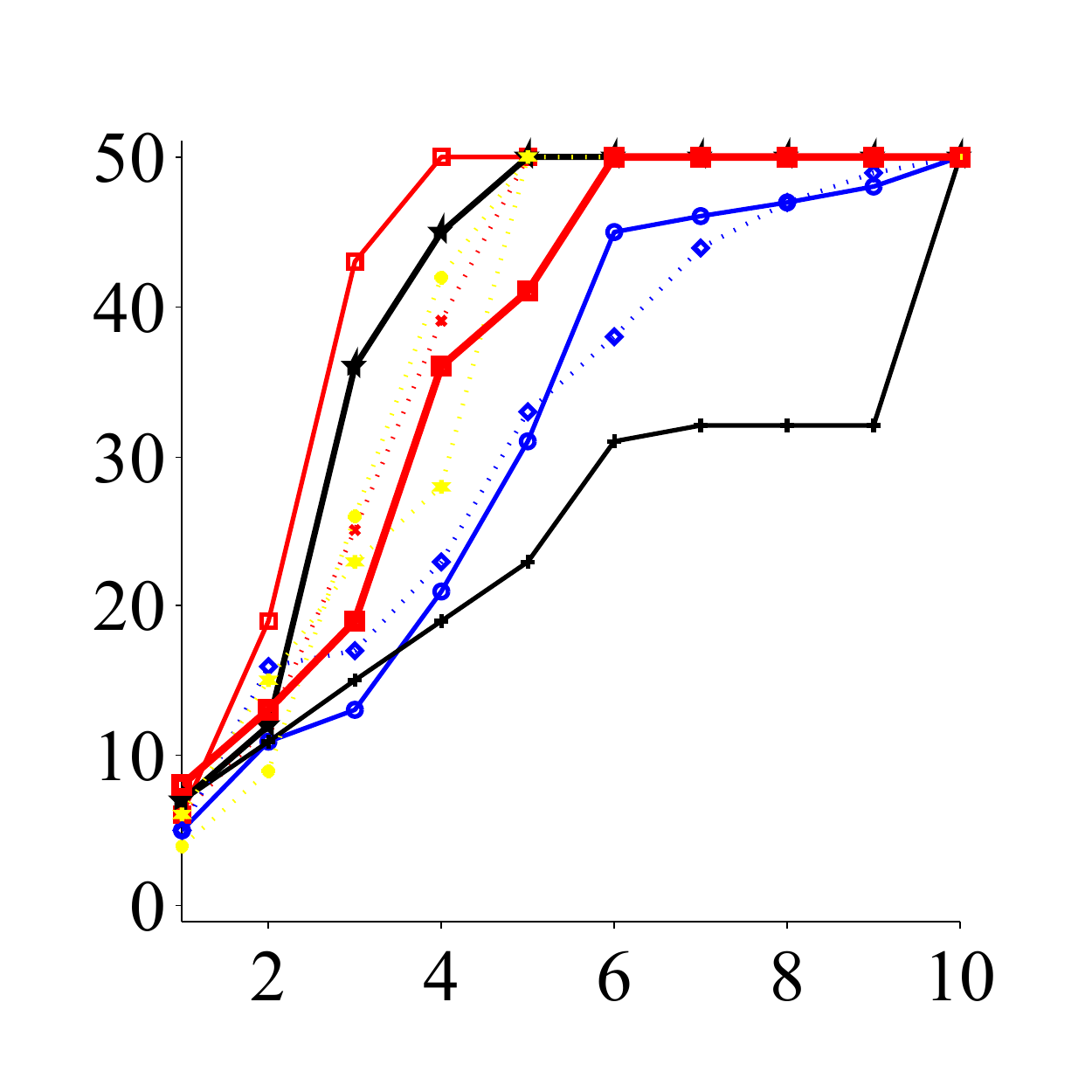}}
	\subfloat[{[Power, 0.5]} ]{\label{fig: power_05_pattern}
		\includegraphics[width=0.16\textwidth]{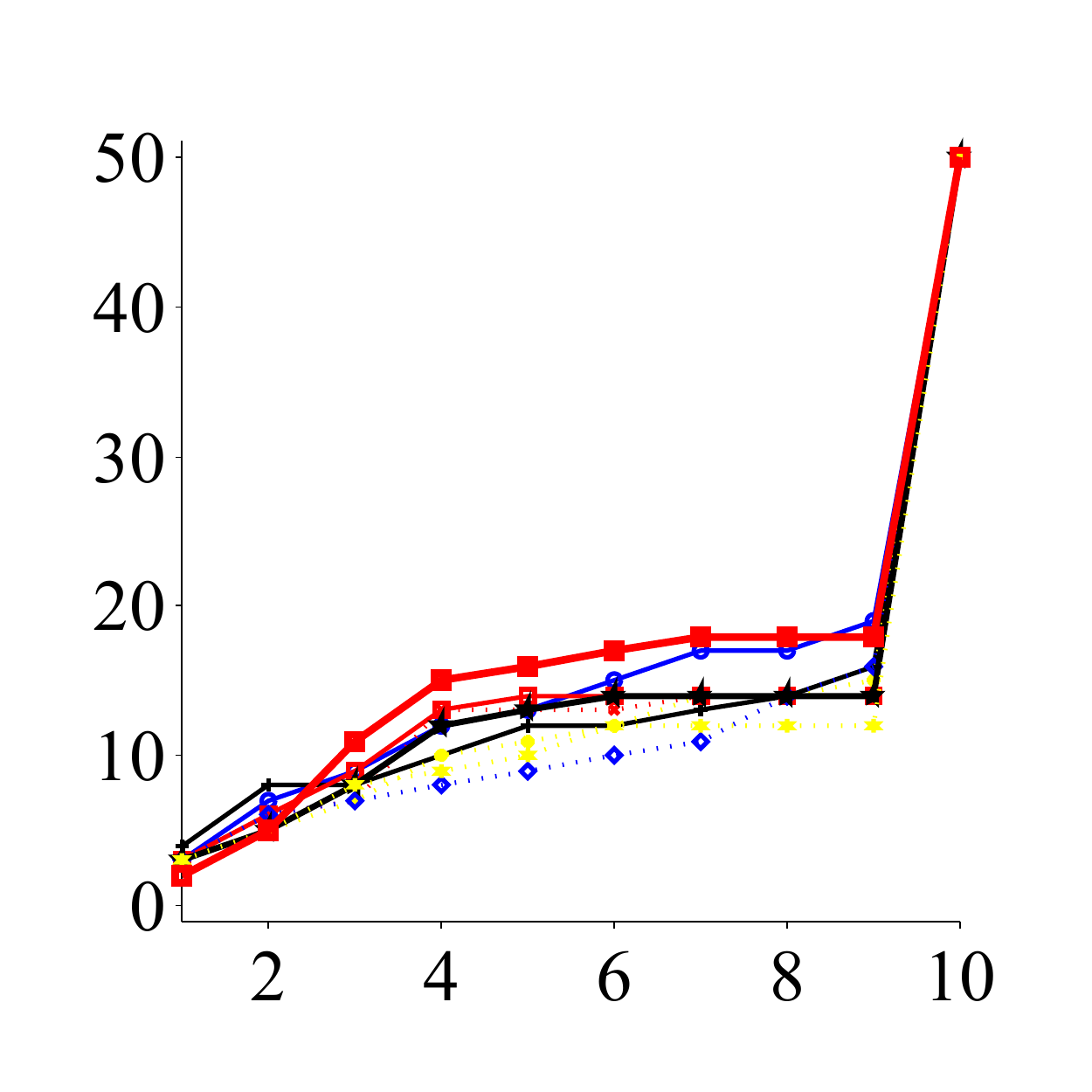}}
	\subfloat[{[Wiki, 0.6]} ]{\label{fig: wiki_06_pattern}
		\includegraphics[width=0.16\textwidth]{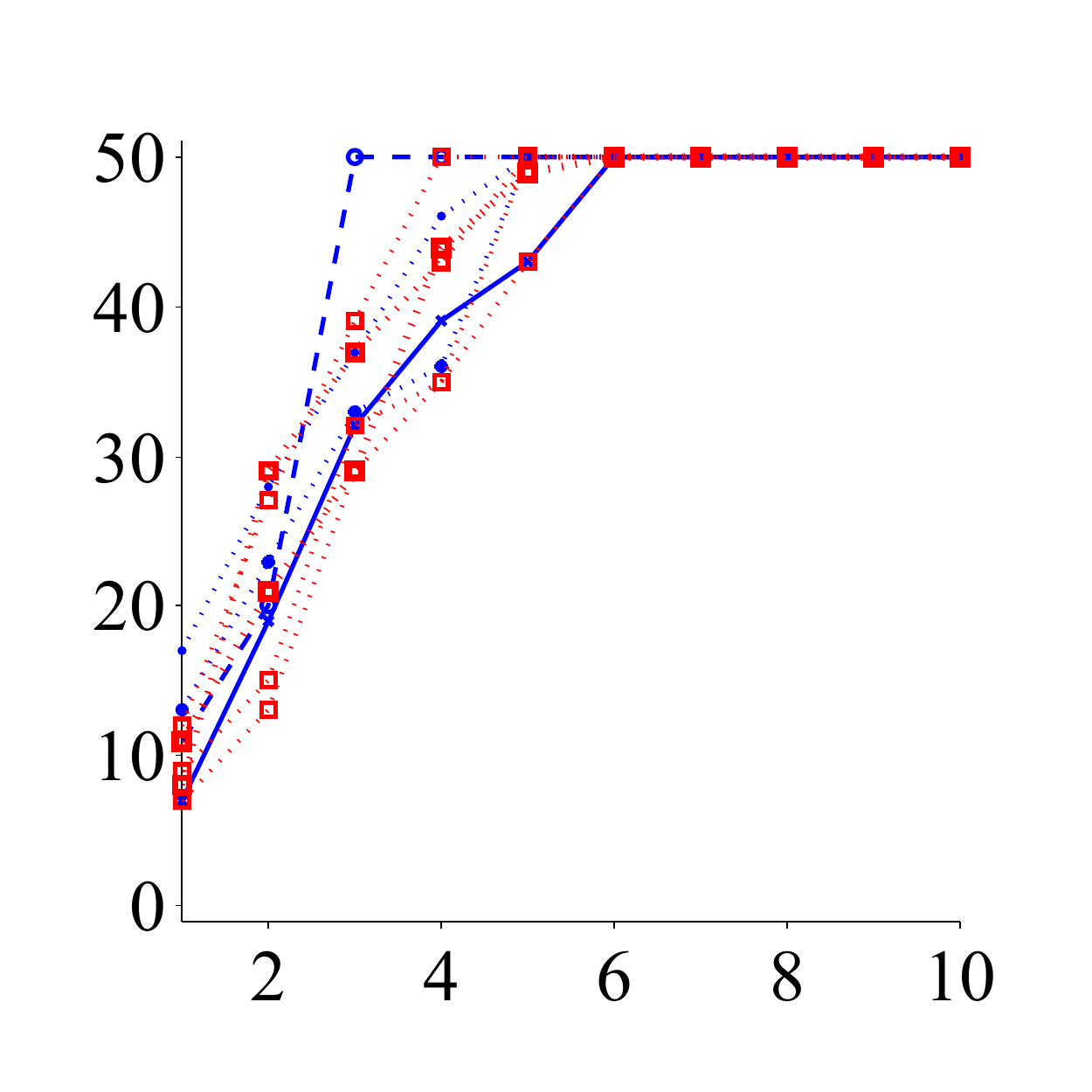}}
	\subfloat[{[Reddit, 0.2]} ]{\label{fig: reddit_02_pattern}
		\includegraphics[width=0.16\textwidth]{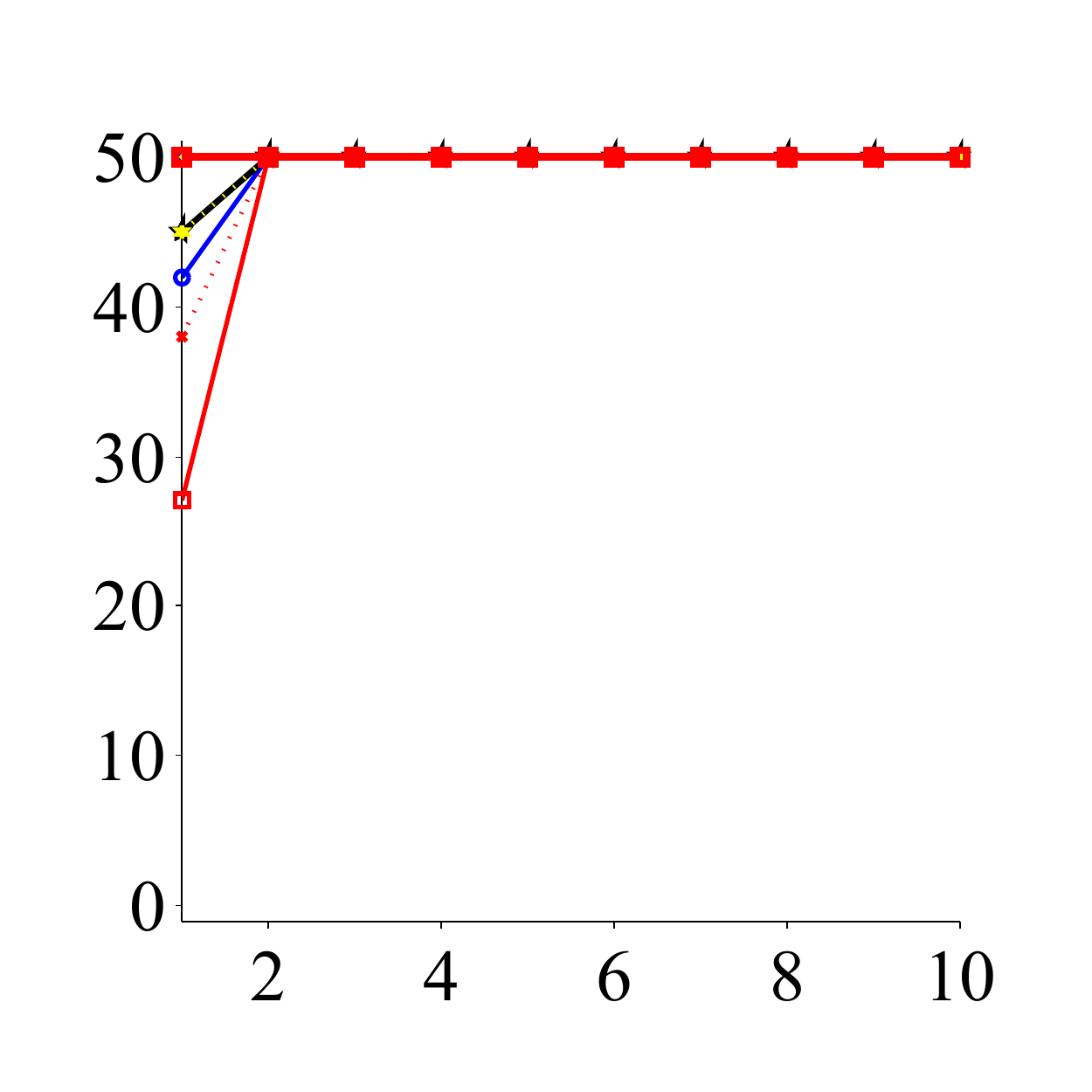}}
	\subfloat[{[Reddit, 0.5]} ]{\label{fig: reddit_05_pattern}
		\includegraphics[width=0.16\textwidth]{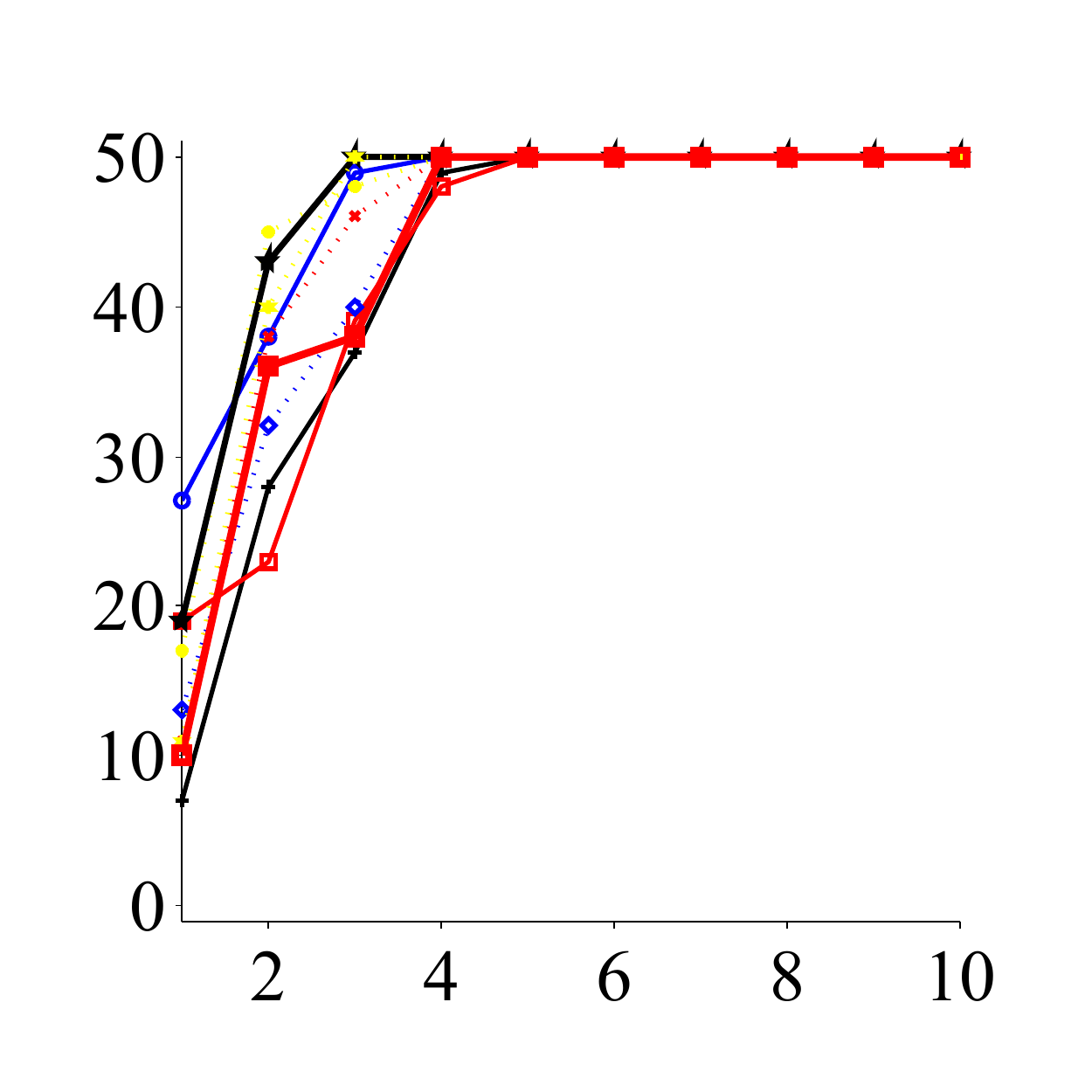}}
	\subfloat[{[Youtube-0.005, 0.6]} ]{\label{fig: youtube_06_pattern}
		\includegraphics[width=0.15\textwidth]{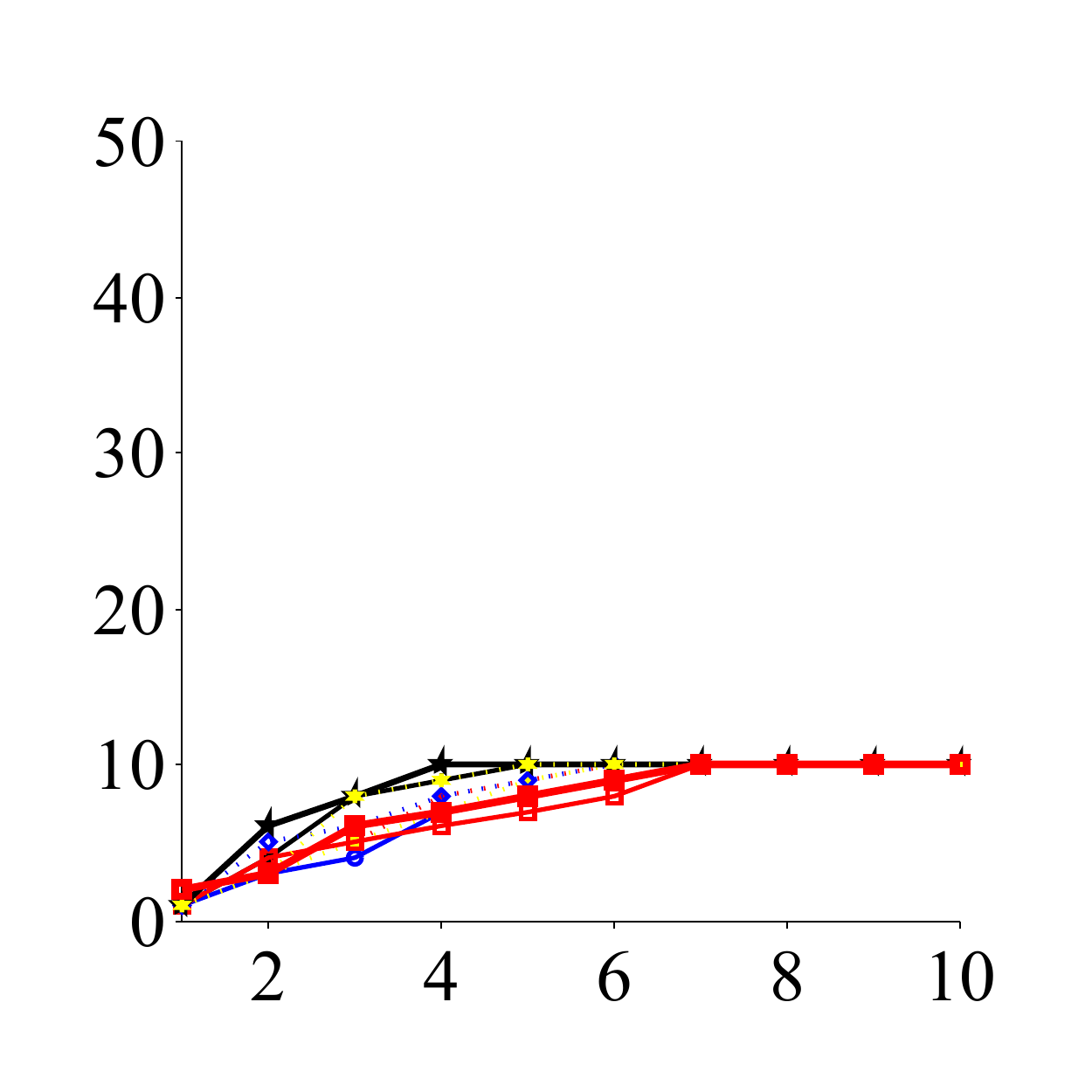}}
	\caption{Cumulative pattern of seed set size. Each subgraph is labeled as {[$\text{dataset}, \theta$]}, and it gives the results of ten random experiments where the point $(x,y)$ shows the total budget $y$ used by the $x$-th seeding step.}
	\label{fig: ff_pattern}
	\vspace{-4mm}
\end{figure*}


\textbf{Analysis of FF Policy.} According to the construction of $\ifadd()$ and $\Ind()$ in Eq. (\ref{eq: ind}), the $\theta$ close to either $0$ or $1$ is not desired, which can be seen from Table. \ref{table: influence}. One interesting observation from Table. \ref{table: influence} is that the performance of FF is concave with respect to $\theta$. For instance, the performance on Reddit is monotone increasing on $[0, a]$ with $a\approx 0.5$ while monotone decreasing after $\theta=a$. We can see that the optimal point varies over different datasets. For example, the best performance is given at $\theta \approx 0.26$ on Power, but for Reddit the optimal point is at $\theta \approx 0.5$. While we do not have any prior estimate on the optimal $\theta$, such a concave pattern suggests that a binary search can be effective.
Second, since the seeding patterns are constructed in real-time, we are interested in that if such patterns are robust. To this ends, for each setting, we plot the patterns generated in ten random simulations, as shown in Fig. \ref{fig: ff_pattern}. As shown in the figure, while the seeding patterns are not exactly the same in different simulations, they do exhibit a similar pattern under the same $\theta$. For example, on Power, most of the budget is used by the $6$-th seeding step under $\theta=0.2$, while more than half of the budget is used after the $8$-th seeding step under $\theta=0.5$.


\textbf{Summary.} Overall, the FF policy is cost-effective in most cases except for the extreme settings, and it results in meaningful and robust seeding patterns controlled by $\theta$. The static policy is worse than FF in average, but it can deal with extreme cases such as Youtube with $(T,K,p_e)=(10,20,0.01)$ or $(T,K,p_e)=(10,10,0.001)$. The SOF policy is effective on small datasets but time-consuming, so reducing its time complexity can potentially make it a desired practical solution for large datasets. Finally, the baselines, Greedy and NonAd, only perform well in certain extreme cases where the TAIM reduces to AIM or TIM.

\section{Conclusion}
\label{sec: con}
In this paper, we have studied the time-constrained adaptive influence maximization (TAIM) problem. The outcomes include the hardness result in computing the optimal policy, a lower bound of the adaptive gap, and, a series of seeding policies. In particular, we show the new hardness in the TAIM problem and observe a critical trade-off for designing effective seeding policies. 


\appendices


\section{{Proof of Lemma \ref{lemma: hardness_2}}}
\label{proof: lemma: hardness_2}
Given a graph and two nodes $s_1$ and $s_2$, the s-t connectedness problem asks for the number of subgraphs in which $s_1$ and $s_2$ are connected. Its decision version is given as follows.
\begin{problem}
	\label{problem: s-t-decision}[\textbf{s-t connectedness}]
	Given a directed graph $G_s=(V_s, E_s)$, an integer $k$ and two nodes $s_1$ and $s_2$, decide whether the number of $s_1$-$s_2$ connected subgraphs is no larger than $k$.
\end{problem}

An oracle of Problem \ref{problem: s-t-decision} can be used to answer the s-t connectedness problem by a binary search, and the oracle is called $O(|E_s|)$ times because the maximum number of s-t connected subgraphs is $2^{|E_s|}$. Since the s-t connectedness problem is \#P-complete \cite{valiant1979complexity}, a polynomial algorithm for Problem \ref{problem: s-t-decision} would yield $NP=P$. Next, we give a reduction from Problem \ref{problem: s-t-decision} to TAIM.

Let us consider an instance of Problem \ref{problem: s-t-decision} given by $(G_s, s_1, s_2, k)$. Let $n_s$ and $m_s$ be the number of nodes and edges in $G_s$, respectively. Without loss of generality, we assume that there is no edge pointing out from $s_2$ and no edge pointing into $s_1$.

\textbf{Reduction.} We construct an instance of TAIM as follows. Let $p_1$ and $p_2$ be small real numbers in $(0,1)$, and $A, B, C$ and $D$ be integers where $C=4n_s$, $D=\frac{1-\frac{k}{2^{m_s}}\cdot p_2}{(1-\frac{k}{2^{m_s}})\cdot p_1\cdot p_2}\cdot C$, $B=4D$ and $A=4B$.\footnote{We omit the rounding issue as it is not critical.} We intent to make the following relationship satisfied
\begin{equation}
\label{eq: reduction}
A\gg B \gg D >C \gg  n_s
\end{equation}
The social network structure $G=(V, E)$ is shown in Fig. \ref{fig: reduction}, built through the following steps:
\begin{itemize}
	\item Copy the graph $G_s$ with $s_1$ and $s_2$. For each edge  $e $ in graph $G_s$, set $p_e$ as $0.5$.
	\item (\textbf{Path $P_s$}) Insert $n_s-1$ new nodes with $n_s$ edges so that the nodes form a simple path from $s_1$ to $s_2$. Let one of the $n_s$ added edges have propagation probability $p_1$ and other edges have propagation probability $1$. We denote this path as $P_s$.
	\item (\textbf{Group A}) Insert $A$ new nodes to the graph and let them be connected from $s_1$. We set that $p_e=1$ for each added edge $e$.
	\item Insert a new node labeled as $s_3$ and an edge $(s_2,s_3)$ with $p_{(s_2,s_3)}=p_2$.
	\item (\textbf{Group B}) Insert $B$ new nodes and let them be connected from $s_2$. We set that $p_e=1$ for each added edge $e$.
	\item (\textbf{Group C}) Select a node inserted in the last step and label it as $s_4$. Insert $C$ new nodes and make them connected from $s_4$. We set that $p_e=1$ for each added edge $e$.
	\item (\textbf{Group D}) Insert another new node $s_5$ and let it connect to $D$ new nodes. We set that $p_e=1$ for each added edge $e$.
\end{itemize}
We set the budget as $K=2$ and the time constraint as $T=n_s+2$. Now the instance of TAIM is completed.

\begin{figure}[t]
	\begin{center}
		\includegraphics[width=0.45\textwidth]{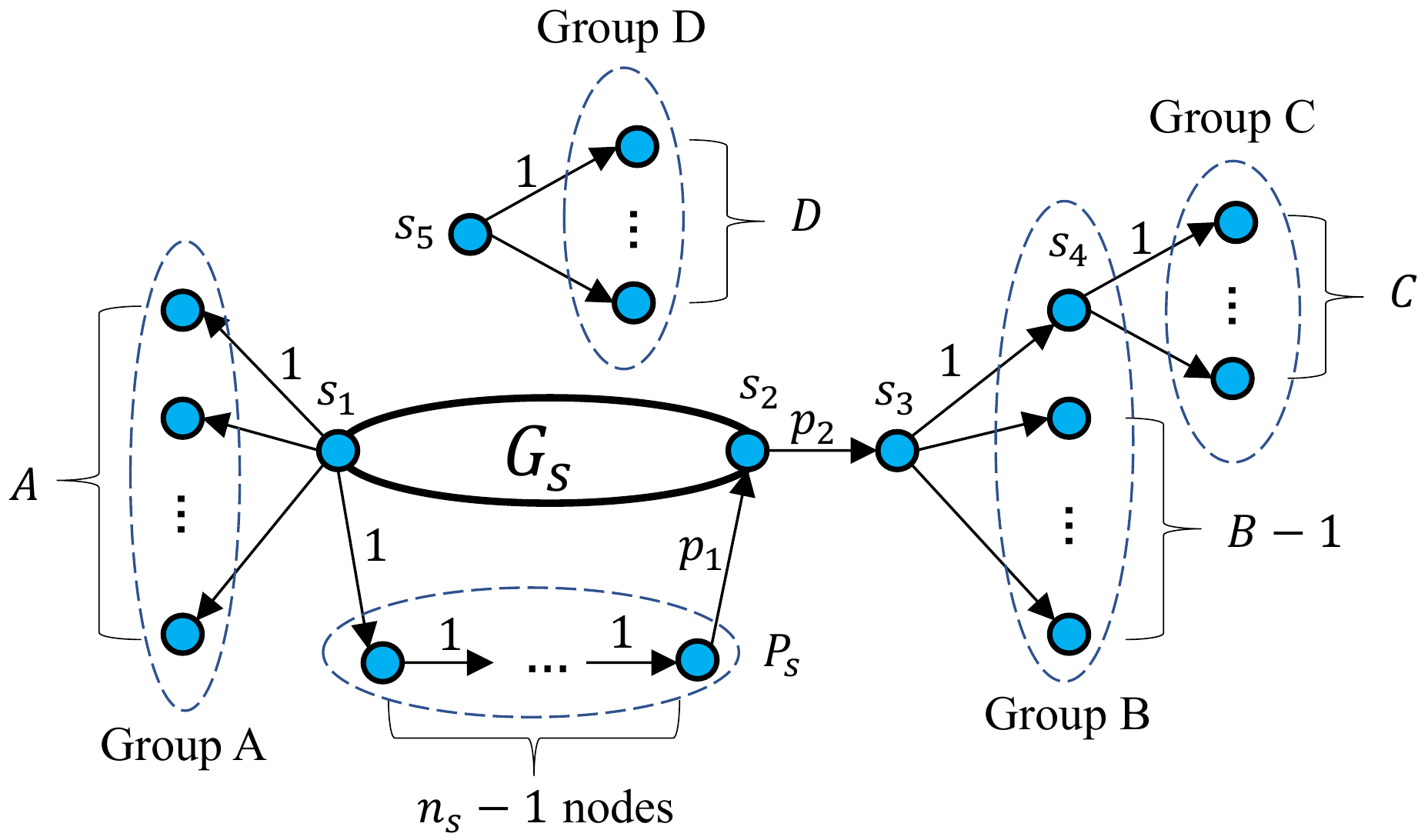} 
	\end{center} 
	\vspace{-2mm}
	\caption{\textbf{Reduction}}
	\label{fig: reduction}
	\vspace{-6mm}
\end{figure}

\textbf{Optimal Policy.} In the optimal policy, to maximize the number of active nodes, only the nodes in $\{s_1,s_3,s_4,s_5\}$ will be selected as seed nodes due to Eq. (\ref{eq: reduction}) as well as the fact that the budget is two and $p_1$ and $p_2$ are small. For the same reason, $s_1$ be must be selected in an optimal policy, so without loss of generality we assume that it is selected as the first seed node in the first seeding step. Now the problem left is to decide when to use the other budget. As we will only select seed nodes from $\{s_3,s_4,s_5\}$ unless they have all been activated, the only event that affects our decision is that if $s_3$ is activated. After each of the first $n_s$ diffusion rounds, once $s_3$ is activated, we should select $s_5$ as the other seed node. If $s_3$ is not activated, we can either wait for more diffusion rounds or select $s_3$ to maximize the number of active nodes. Because the time constraint is $n_s+2$ and leaving two diffusion rounds is sufficient for $s_3$ to activate all the nodes connected from it, it is optimal to wait until the $n_s$-th diffusion round. 

After the $n_s$-th diffusion round, if $s_3$ has been activated, it is clear that we should select $s_5$ as the second seed node. If $s_3$ is not activated yet, we would have two choices: (a) select the second seed node or (b) wait for another diffusion round and then select the second seed node. Note that there are only two diffusion rounds left, so the optimal policy must be one of those choices. Now we calculate the resulted influence.

\textbf{Policy a.} Suppose the seed node must be selected right after $n_s$ diffusion rounds. In this case, the profit is 
\begin{equation}
A+E_s+B+\Pr[\leq n_s]\cdot  (C+D)+(1-\Pr[\leq n_s])\cdot C+O(1)
\end{equation}
where $\Pr[\leq n_s]$ is probability that $s_3$ can be activated within $n_s$ rounds of diffusion, and $E_s$ is expected number of active nodes in $G_s \cup P_s$ resulted from $s_1$.

\textbf{Policy b.} We would wait for another diffusion round even if $s_3$ is not activated after $n_s$ diffusion rounds. Since the time constraint is $n_s+2$, this seed node must selected right after the $n_s+1$ diffusion round. Under this policy, when $s_3$ is activated before the $(n_s+1)$-th diffusion round, we would select $s_5$ as the second seed node, and therefore the total profit is $A+E_s+B+C+D+2$. If $s_3$ is activated exactly in the $(n_s+1)$-th diffusion round, we should select $s_5$ as the seed node, because there is only round left and $D>C$. In this case, the total profit is $A+E_s+B+D+2$. If $s_3$ is not activated after the $(n_s+1)$-th diffusion round, we should select $s_3$ to maximize the profit as $B$ is larger than $C$ or $D$, and therefore the total profit is $A+E_s+B+1$. In summary, the total profit under the second policy is 
\begin{equation}
A+E_s+B+ \Pr[\leq n_s]\cdot (C+D)+\Pr[= n_s+1]\cdot D+O(1). 
\end{equation}
where $\Pr[=n_s+1]$ is the probability that $s_3$ is activated exactly after $n_s+1$ rounds of diffusion.

Comparing the above two policies, for sufficiently large $n_s$, Policy a is better than Policy b if and only if $\frac{1-\Pr[\leq n_s]}{\Pr[=n_s+1]}\geq \frac{D}{C}$. Let $p^*$ be the probability that $s_2$ can be activated through $G_s$ from $s_1$. Because the longest simple path from $s_1$ to $s_2$ in $G_s$ has at most $n_s-1$ edges and the path $P_s$ has $n_s$ edges, $\Pr[\leq n_s]$ is the probability that $s_2$ is first activated by $s_1$ through $G_s$ but not $P_s$, and then $s_3$ is activated by $s_2$, which means $\Pr[\leq n_s]=p^*\cdot p_2$. Similarly, $\Pr[= n_{s}+1]$ is the probability that $s_2$ is first activated by $s_1$ through the path $P_s$ but not $G_s$, and $s_3$ is then activated by $s_2$, implying that $\Pr[= n_{s}+1]=(1-p^*)\cdot p_1 \cdot p_2$. Therefore, Policy a is better than Policy b if and only if $\frac{1-p^*\cdot p_2}{(1-p^*)\cdot p_1\cdot p_2}\geq \frac{D}{C} \iff p^* \leq \frac{k}{2^{m_s}}$. Because the probability of edges in $G_s$ is uniformly 0.5, $p^*$ is equal to $\frac{n^*}{2^{m_s}}$ where $n^*$ is the number of $s_1$-$s_2$ connected subgraphs. Thus, deciding which policy is better is equivalent to determining if the number of $s_1$-$s_2$-connected subgraphs is no larger than $k$, which completes the proof.

\bibliography{bib_amo}
\bibliographystyle{IEEEtran}

\begin{IEEEbiography}[{\includegraphics[width=1.0in,clip,keepaspectratio]{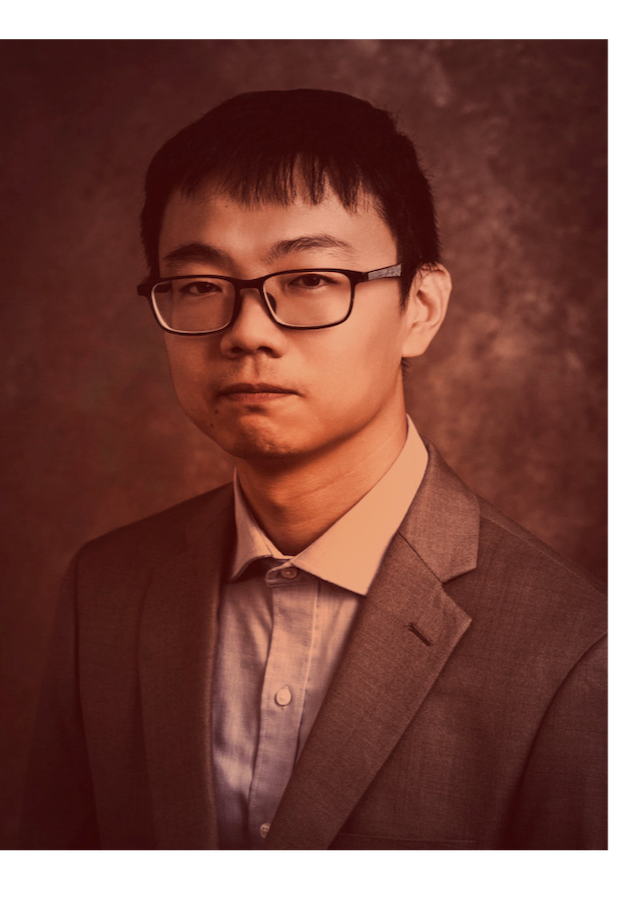}}]{Guangmo Tong (M'18)} is an Assistant Professor in the Department of Computer and Information Sciences at the University of Delaware. He received a Ph.D. in Computer Science at the University of Texas at Dallas in 2018. He received his BS degree in Mathematics and Applied Mathematics from Beijing Institute of Technology in July 2013. His research interests include computational social systems, machine learning, and theoretical computer science. He has published articals in vairous journals and conferences such as IEEE Transaction on Networking, IEEE Transaction on Computational Social System, IEEE INFOCOM, and NeurIPS. 
\end{IEEEbiography}
\begin{IEEEbiography}[{\includegraphics[width=0.9in,clip,keepaspectratio]{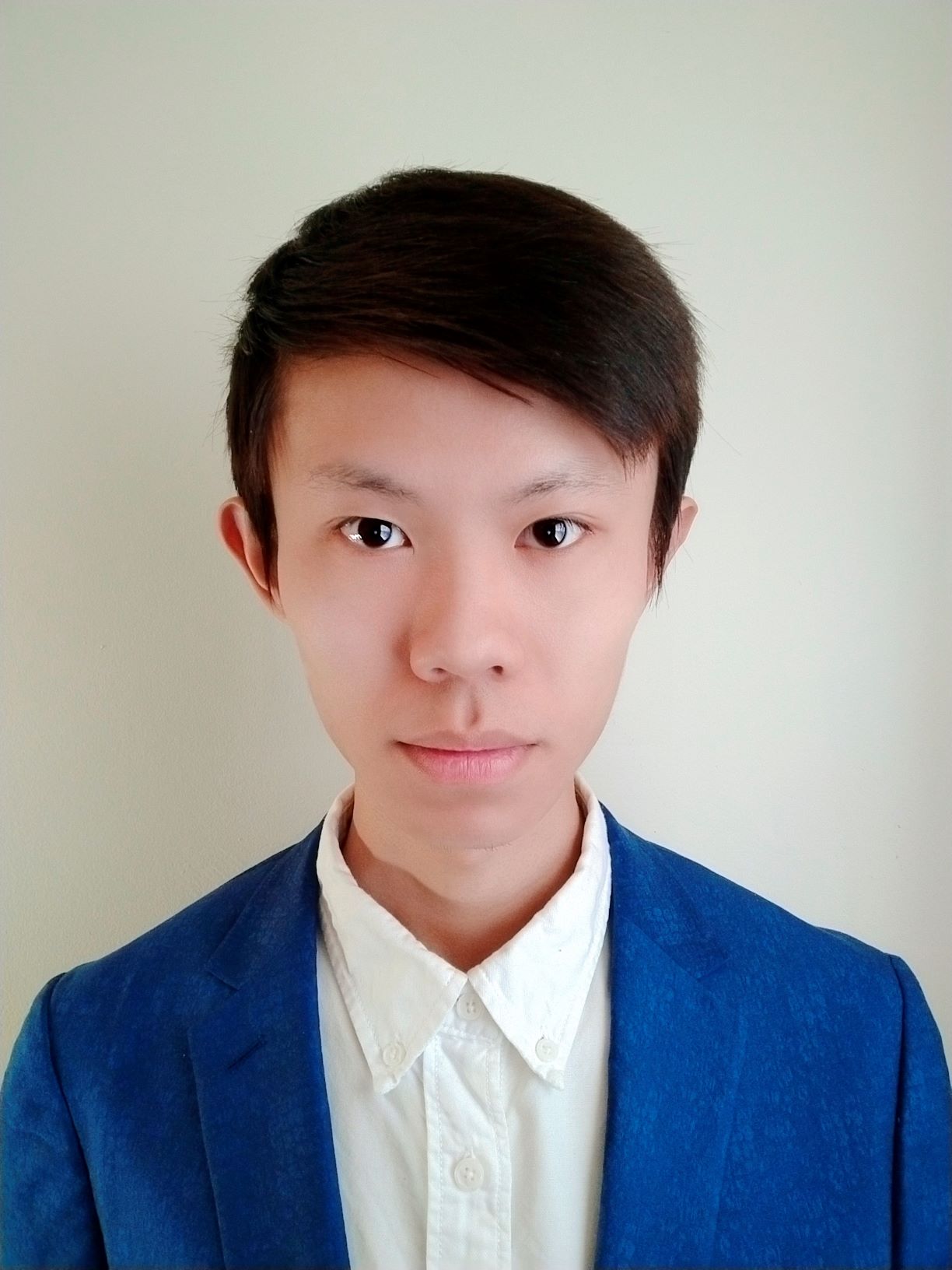}}]{Ruiqi Wang} received his B.E. degree in Information and Software Engineering from University of Electronic Science and Technology of China, in 2018. He is currently pursuing a Master degree in Computer Science at University of Delaware. His current research interests are in the area of social networks and information diffusion.
\end{IEEEbiography}
\begin{IEEEbiography}[{\includegraphics[width=1.0in,clip,keepaspectratio]{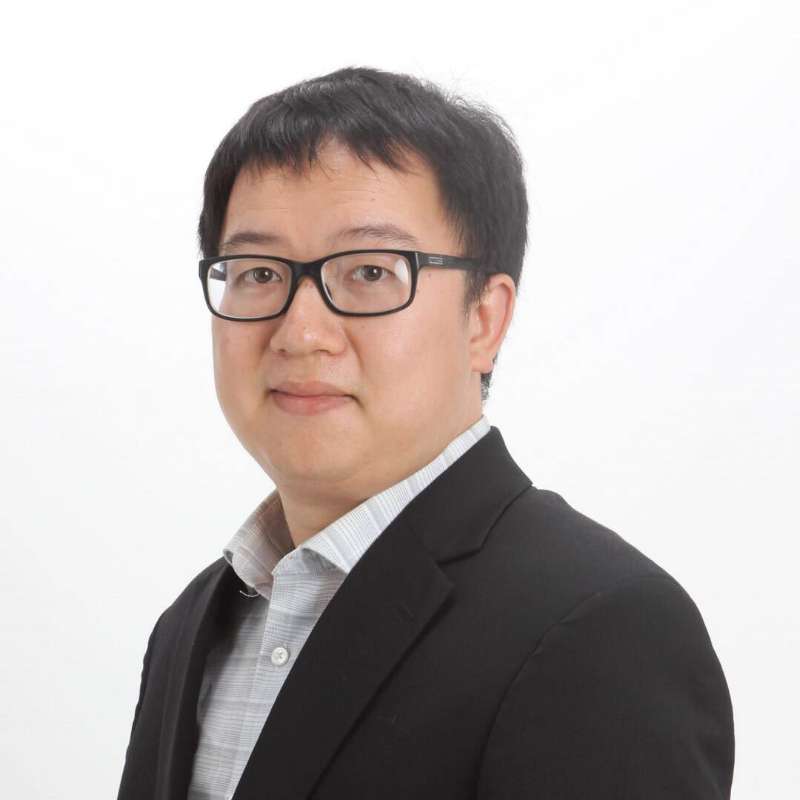}}]{Zheng Dong (M'19)} is an assistant professor in the Department of Computer Science at Wayne State University. He received the PhD degree from the Department of Computer Science at the University of Texas at Dallas in 2019. His research interests include real-time cyber physical systems and mobile edge computing. He received the Outstanding Paper Award at the 38th IEEE RTSS. He is a member of the IEEE.
\end{IEEEbiography}

\begin{IEEEbiography}[{\includegraphics[width=1.in,clip,keepaspectratio]{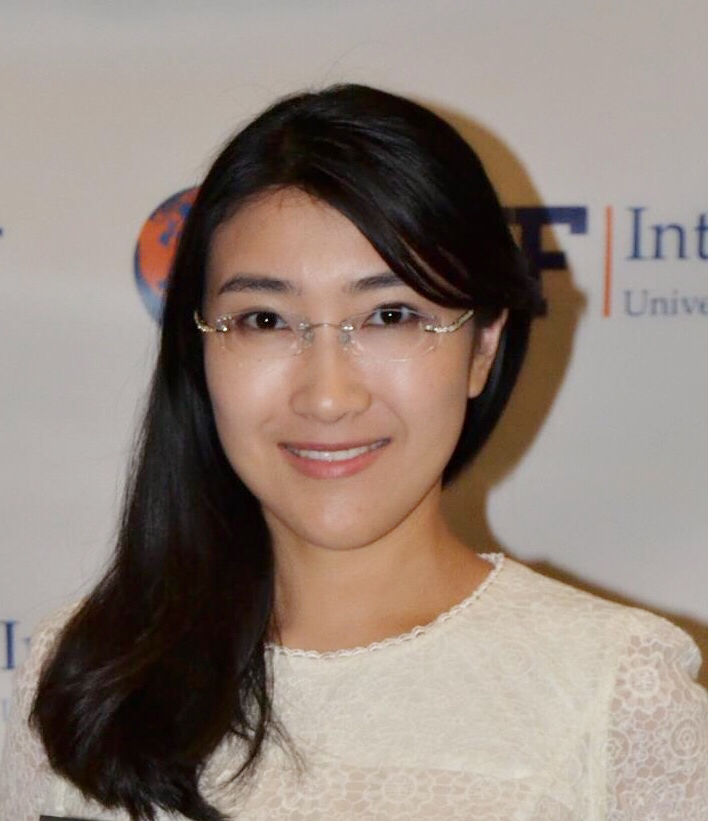}}]{Xiang Li (M’18)} is an Assistant Professor at the Department of Computer Engineering of Santa
Clara University. She received her Ph.D. degree in Computer and Information Science and Engineering
department of the University of Florida. Her research interests are centered on the large-scale
optimization and its intersection with cyber-security of networking systems, big data analysis, and cyber
physical systems. She has published 25 articles in various prestigious journals and conferences such
as IEEE Transactions on Mobile Computing, IEEE Transactions on Smart Grids, IEEE INFOCOM, IEEE ICDM, including one Best Paper Award in IEEE MSN 2014, Best Paper Nominee in IEEE ICDCS 2017, and Best Paper Award in IEEE International Symposium on Security and Privacy in Social Networks and Big Data 2018.
She has served as Publicity Co-Chair of International Conference on Computational Data \& Social Networks 2018, Session Chair of ACM SIGMETRICS International Workshop 2018, and on TPC of many conference
including IEEE ICDCS, IEEE ICDM workshop, COCOA etc., and also served as a reviewer for several journals such as IEEE Transactions on Mobile Computing, IEEE Transactions on Networks Science and Engineering and Journal of Combinatorial Optimization, etc.
\end{IEEEbiography}
\vfill

\newpage

\ifCLASSOPTIONcaptionsoff
  \newpage
\fi

\end{document}